\documentclass{article}

\usepackage{fullpage}
\usepackage{amsmath,amsthm,amssymb}
\usepackage{floatpag}
\usepackage[ruled,titlenotnumbered,linesnumbered,noline,noend]{algorithm2e}
\usepackage[dvipsnames,usenames]{color}
\usepackage{graphicx}
\usepackage{epstopdf}

\theoremstyle{definition}

\newtheorem{theorem}{Theorem}[section]
\newtheorem{corollary}{Corollary}[section]
\newenvironment{proofof}[1]{\noindent {\em Proof of #1.  }}{\hfill$\Box$}
\newtheorem{definition}{Definition}[section]
\newtheorem{lemma}{Lemma}[section]
\newtheorem{claim}{Claim}[section]
\newtheorem{observation}{Observation}[section]
\def \Z {{\mathbb Z}}
\newcommand{\ignore}[1]{}

\newcommand{\cupdot}{\mathbin{\mathaccent\cdot\cup}}

\newcommand{\etal}{{\em et al.\ }}

\title{Optimal Distributed Coloring Algorithms for Planar Graphs in the LOCAL model}

\author{
Shiri Chechik
\thanks{Tel Aviv University, Israel.
E-mail: {\tt shiri.chechik@gmail.com}.
}
\and
Doron Mukhtar
\thanks{Tel Aviv University, Israel.
E-mail: {\tt doron.muk@gmail.com}.
}
}

\begin{document}

\begin{titlepage}
\maketitle
\thispagestyle{empty}

\begin{abstract}
In this paper, we consider distributed coloring for planar graphs with a small number of colors.
We present an optimal (up to a constant factor) $O(\log{n})$ time algorithm for 6-coloring planar graphs.
Our algorithm is based on a novel technique that in a nutshell detects small structures that
can be easily colored given a proper coloring of the rest of the vertices and
removes them from the graph until the graph contains a small enough number of edges.
We believe this technique might be of independent interest.

In addition, we present a lower bound for 4-coloring planar graphs that essentially shows that any algorithm (deterministic or randomized) for $4$-coloring planar graphs requires $\Omega(n)$ rounds.
We therefore completely resolve the problems of 4-coloring and 6-coloring for planar graphs in the LOCAL model.

\end{abstract}
\end{titlepage}

\section{Introduction}

Given a graph $G = (V,E)$ and a positive integer $k$, the $k$-coloring problem asks for an assignment of colors (integers from $\{1,...,k\}$) to the vertices of $G$ such that no two adjacent vertices in $G$ are assigned the same color. We study this problem in the distributed LOCAL model of computation. In this model, a communication network is represented by a graph where each vertex corresponds to a processor and each edge to a communication line. The vertices communicate by sending messages over these edges in a synchronous manner, that is, they communicate in discrete rounds. In each such round, each vertex receives the messages that were sent to it in the previous round, preforms some local computation and sends messages to some of its neighbors (where there is no limit on the size of the messages and on the local computation time).
The running time of an algorithm in this model is equal to the number of communication rounds that the algorithm requires in the worst case.
The CONGEST model is very similar to the LOCAL model but the size of the messages is bounded (usually by some small poly-log).
Our algorithms can also be tweaked to work in the CONGEST model.

Vertex coloring is one of the most fundamental and well-studied problems in the area of distributed computing with a long history dating back to the 80's. Algorithms for vertex coloring have numerous applications in the design of other distributed algorithms, scheduling in networks and many other problems.
Perhaps the most common form of distributed coloring is the $(\Delta+1)$-coloring problem\footnote{As usual, $\Delta$ is the maximum degree in the graph}.
In the centralized setting, computing $(\Delta+1)$-coloring can be easily done in linear time, by simply iterating over the vertices in an arbitrary order and coloring the current vertex with a color not yet used by one of its neighbors (it should be noted that such a color must exist as the maximum degree is $\Delta$).
Computing $(\Delta+1)$-coloring in the distributed setting is much trickier.
In the deterministic approach, several algorithms with running
time of $O(f(\Delta) + \log^*n)$ have been developed \cite{BarenboimEK14,Linial90,Goldberg87,GoPl87} culminating with sublinear in $\Delta$ time algorithms (see recent breakthrough results \cite{Barenboim15,FrHeKo15})\footnote{As usual, $n$ (respectively, $m$) is the number of vertices (resp., edges) in the graph.}.
The $\log^*n$ term is necessary as Linial \cite{Linial90} showed that even 3-coloring a ring requires $\Omega(\log^*n)$ round.
The situation for the randomized algorithms is much more optimistic.
The first randomized $(\Delta+1)$-coloring can be tracked back to the $O(\log{n})$ MIS \cite{Luby85,ABI86}.
This trigged a series of works culminating with the recent breakthroughs of
Harris, Schneider, and Su \cite{HaScSu16} who presented a randomized $(\Delta+1)$-coloring with $O(\sqrt{\log{\Delta}} + 2^{\sqrt{\log\log{n}}})$ time, and Chang, Li and Pettie \cite{Pet15} who presented a randomized $(\Delta+1)$-coloring with $O(\log^*n + \text{Det}_d(\text{poly} \log n))$ time (where $\text{Det}_d(n')$ is the deterministic complexity of (deg+1)-list coloring on $n'$-vertex graphs).

Obtaining coloring with fewer colors is a desirable property that is crucial to the efficiency of many applications that use coloring.
All the above mentioned algorithms use $(\Delta+1)$-coloring even if $\Delta$ is very large and even if it is possible to use fewer colors.
Much less is known for general graphs when we restrict the algorithm to use less than $\Delta$ colors.

In this paper, we focus on distributed coloring planar graphs with a small number of colors (independent of $\Delta$).
In the sequential centralized setting it is well known that a planar graph can be colored using 4 colors in polynomial time (e.g. \cite{appel1989}), however no near linear time algorithm is known for this problem.
If we allow 5 colors then linear time algorithms exist (see e.g. \cite{Chiba81}).

Coloring planar graphs was also studied in the PRAM model, where the computation is partitioned among polynomial number of processors and the processors can communicate through a shared memory.
Goldberg \etal \cite{Goldberg87} presented a deterministic $O(\log{n})$ time algorithm for $7$-coloring planar graphs in the PRAM model of computation.
Although the use of PRAM model, this algorithm can be easily converted to an algorithm that runs in a distributed model of computation.
However, this approach fails when restricting the algorithm to use less than 7 colors, as it strongly relies on the fact that at least a constant fraction of the vertices are of degree at most 6 in any planar graph.
Barenboim and Elkin \cite{Barenboim08} generalized this algorithm to graphs of bounded arboricity and devised a deterministic algorithm for $\left(\lfloor (2 + \varepsilon) \cdot a \rfloor + 1\right)$-coloring graphs of arboricity at most $a$ that has running time of $O(a\log{n})$ (where the parameter $\varepsilon$ is an arbitrary small positive constant).

Goldberg \etal \cite{Goldberg87} also presented a different deterministic algorithm for 5-coloring planar graphs (where the embedding of the graph is given) in the PRAM model. This algorithm finds a $5$-coloring in $O(\log{n}\log^{*}{n})$ time. In contrast to the previous algorithm, this algorithm is based on iterative contraction of vertices, and so requires the maintenance of a virtual graph composing of super nodes where each of them corresponds to a contraction of several vertices in the original graph that might be at distance $\Omega(n)$ from each other. In a distributed model of computation, maintaining such a graph requires $\Omega(n)$ rounds.
Hagerup \etal \cite{Opt5} gave a similar algorithm for $5$-coloring planar graphs that does not require an embedding of the graph to be given as an input.
However, this algorithm is also based on vertex contraction.
Krzysztof  \cite{Krzysztof} gave another $6$-coloring algorithm in the PRAM model. However, it also requires $\Omega(n)$ time in a distributed model of computation.

If we want to use less than $7$ colors in the distributed setting for planar graphs, it seems that no non-trivial (that requires $o(n)$ time) algorithm is known.


A well known result by Linial \cite{Linial90} states that coloring trees with a constant number of colors requires time $\Omega(\log n)$. Since trees are planar graphs, it follows that any algorithm that finds a coloring of planar graphs with a constant number of colors requires $\Omega(\log n)$ time.

In this paper we present an optimal (up to constant factor) algorithm for 6-coloring planar graphs (and with less colors for related families). More specifically, we obtain a deterministic $O(\log{n})$ time algorithm for 6-coloring planar graphs, which matches the lower bound of Linial \cite{Linial90}.
Our algorithm is based on a novel technique that removes small structures (that can be later easily colored) from the graph until the graph contains a small number of edges.
We believe that our techniques may be of independent interest.
In addition, we present a lower bound for 4-coloring planar graphs that essentially shows that any algorithm (deterministic or randomized) for $4$-coloring planar graphs requires $\Omega(n)$ rounds.
We therefore completely resolve the problems of 4-coloring and 6-coloring for planar graphs in the LOCAL model.

\section{Preliminaries} \label{sec:preliminaries}
     We denote by $\Z_{\ge}$ the set of all non-negative integers and by $\Z_{>}$ the set of all positive integers. For a graph $H$, we respectively denote by $V(H)$ and $E(H)$ the vertex set and edge set of $H$. The degree of a vertex $v$ in a graph $H$ is denoted by $d(v,H)$. Given a graph $H$ and a set $U \subseteq V(H)$, we denote by $H[U]$ the induced graph on $U$. A simple cycle in a graph $H$ is a sequence $(v_{1},...,v_{k})$ of $k > 2$ distinct vertices from $V(H)$ such that $\{v_{i},v_{i+1}\} \in E(H)$ for all $i \in \{1,...,k-1\}$ and $\{v_{k},v_{1}\} \in E(H)$. Given a simple cycle $C = (v_{1},...,v_{k})$ in a graph $H$, we respectively denote by $V(C)$ and $E(C)$ the vertex set and edge set of $C$ (i.e. $V(C) = \{v_{1},...,v_{k}\}$ and $E(C) = \{\{v_{i},v_{i+1}\} \mid i \in \{1,...,k-1\}\} \cup \{\{v_{k},v_{1}\}\}$). The length of $C$, denoted by $|C|$, is defined to be $|V(C)|$, and its degree in $H$, denoted by $\Delta(C,H)$, is defined to be $\max_{v \in V(C)}{d(v,H)}$.

For a positive integer $k$, a proper $k$-coloring of a graph $H$ is a mapping $\varphi$ from the set of vertices $V(H)$ to the set of colors $\{1, ..., k\}$ such that no two adjacent vertices in $H$ are assigned the same color. A partial proper $k$-coloring of a graph $H$ is a mapping $\varphi$ from the set of vertices $V(H)$ to the set $\{1, ..., k\} \cup \{\perp\}$ such that the restriction of $\varphi$ to the set of colored vertices $U = \{v \in V(H) \mid \varphi(v) \ne \; \perp\}$ is a proper $k$-coloring of $H[U]$.

\medskip

We will need Brooks' Theorem \cite{Brooks41}:

\begin{theorem}[Brooks' Theorem] Every connected graph $H$ with maximum degree $\Delta$ has a proper $\Delta$-coloring unless $H$ is either isomorphic to a complete graph or to a cycle graph of odd length.\end{theorem}

We will also need the following well known claim, which is a folklore (see e.g. \cite{BarenboimElkinBook13} for a proof):
\begin{claim} For any integer $g \ge 3$ and any planar graph $H$, if $H$ does not contain any simple cycle $C$ of length $g > |C|$ then $|E(H)| \le \frac{g}{g-2}|V(H)|$.\label{claim:numberOfEdges}
\end{claim}

The rest of the paper is organized as follows.
In Section \ref{sec:removable} we present the notion of removable cycle that is crucial for our algorithm.
In Section \ref{sec:4-coloring}, as a warmup, we  present an $O(\log{n})$ time algorithm for 4-coloring  of triangle-free planar graphs.
In Section \ref{sec:6-coloring}, we present our main result of $O(\log{n})$ time algorithm for 6-coloring  of planar graphs.
Finally, in Section \ref{sec:lowerBounds} we present our lower bound technique showing that any 4-coloring algorithm of planar graphs requires $\Omega(n)$ time, and any 3-coloring algorithm of outerplanar graphs requires $\Omega(n)$ time.

\section{Removable Cycles}
\label{sec:removable}
\begin{definition}[Removable Cycle] A simple cycle $C$ in a graph $H$ is {\em removable} if $H[V(C)]$ is neither isomorphic to a complete graph nor to a cycle graph of odd length.\end{definition}

The next lemma shows that if $C$ is a removable cycle in a graph $H$ with $\Delta(C,H) \le \Delta$, and $\varphi$ is a partial proper $\Delta$-coloring of $H$ in which $C$ is not colored, then one can $\Delta$-color the vertices of $C$ in such a way that this extended coloring remains a partial proper $\Delta$-coloring of $H$.

\begin{lemma}\label{lem:coloringRemovableCycle} For any graph $H$, any removable cycle $C$ in $H$ with $\Delta(C,H) \le \Delta$ and any partial proper $\Delta$-coloring $\varphi$ of $H$ such that $\varphi(v) = \; \perp$ for all $v \in V(C)$, there exists a partial proper $\Delta$-coloring $\varphi'$ of $H$ such that $\varphi'(u) \neq \; \perp$ for all $u \in V(C)$ and $\varphi'(u) = \varphi(u)$ for all $u \in V(H) \setminus V(C)$.\end{lemma}

\begin{proof} Let $H$ and $C = (v_{1},...,v_{k})$ be some removable cycle in $H$ with $\Delta(C,H) \le \Delta$. Let $\varphi$ be a partial proper $\Delta$-coloring of $H$ such that $\varphi(v) = \; \perp$ for all $v \in V(C)$. For each $v \in V(C)$, we denote by $\text{COL}(v)$ the set that contains the color of each colored neighbor of $v$ (that is, $\text{COL}(v) = \{\varphi(u) \mid \{u,v\} \in E(H) $ and $\varphi(u) \ne \; \perp\}$), and by $S(v)$ the vertex that comes after $v$ in the cyclic order of $C$ (i.e. $S(v_i) = v_{i+1}$ for all $i \in \{1, ..., k-1\}$ and $S(v_k) = v_1$). If $\text{COL}(v) \subseteq \text{COL}(S(v))$ for all $v \in V(C)$, then we must have $\text{COL}(v_1) \subseteq \text{COL}(v_2) \subseteq ... \subseteq \text{COL}(v_k) \subseteq \text{COL}(v_1)$ and so there exists a set of colors $\text{COL}$ such that $\text{COL}(v) = \text{COL}$ for all $v \in V(C)$. Since each vertex $v \in V(C)$ must have at least one colored neighbor for each color in $\text{COL}$, we get that each vertex $v \in V(C)$ has at least $|\text{COL}|$ colored neighbors.
Since the vertices in $V(C)$ are not colored, it must be that each vertex $v \in V(C)$ has at most $d(v,H) - |\text{COL}|$ neighbors in $H[V(C)]$. Since $C$ is a removable cycle in $H$, the graph $H[V(C)]$ is neither isomorphic to a complete graph nor to a cycle graph of odd length. Moreover, we have $d(v,H) \le \Delta$ for all $v \in V(C)$, and so the maximum degree of $H[V(C)]$ is at most $\Delta - |\text{COL}|$. It follows from Brooks' theorem \cite{Brooks41} that $H[V(C)]$ can be properly colored using at most $\Delta - |\text{COL}|$ colors and so we can properly color it by using the palette $\{1,...,\Delta\}\setminus \text{COL}$ which is a proper extension of $\varphi$.

Now, if there exists a vertex $v \in V(C)$ for which $\text{COL}(v) \not\subseteq \text{COL}(S(v))$, then there exists a color $x$ such that $x \in \text{COL}(v)$ and $x \not\in\text{COL}(S(v))$. Let $S_1(v) = S(v)$, and $S_i(v) = S(S_{i-1}(v))$ for each $i \in \{2, ..., k\}$ (note that $v = S_k(v)$). We start by assigning $S_1(v)$ the color $x$ (which is a proper extension of $\varphi$ as $x \notin \text{COL}(S(v))$). Next, we properly color each of the vertices in $S_2(v), ..., S_{k-1}(v)$ (in this order) by finding each time a color that no neighbor of the current vertex use (always possible since each of these vertices has at least one uncolored vertex at the moment of its coloration).
Finally, we are left with coloring $S_{k}(v)$, that is, $v$ itself.
Note that $S_{k}(v)$ currently has at least two neighbors colored $x$ ($x\in \text{COL}$ so $v$ has at least one neighbor not in $V(C)$ colored $x$ and in addition we colored
 $S_1(v)$ in $x$ as well).
Therefore, there is a free color (that is, a color that no neighbor of $v$ uses) to assign to $v$.
 \end{proof}

 \section{Finding a 4-coloring of triangle-free planar graphs}
\label{sec:4-coloring}

We assume that the communication network $G = (V,E)$ is a triangle-free planar graph on $n$ vertices where each vertex $v \in V$ has a unique identifier $\text{ID}(v)$ taken from $\{1,...,n^{c'}\}$ (for some constant integer $c' \ge 1$). At the beginning each vertex $v \in V$ knows only the value of $n$ and its own identifier.

The very rough idea of the algorithm is to partition the vertices of the graph into $O(\log n)$ disjoint subsets that will later be $4$-colored iteratively one after the other.
More precisely, the algorithm partitions the set of vertices $V$ into sets of vertices $V_1,...,V_{r}$ (for some $r = O(\log{n})$) such that the following occurs (in our algorithm below these sets are the sets with vertices $u$ with the same label $\text{L}[u]$).
For every $i$, assuming the vertices in the set $\cup_{r\geq j>i}{V_{j}}$ are already properly colored (and no other vertex is yet colored), we can efficiently assign colors to the set of vertices $V_i$ such that the previous coloring together with the new coloring forms a proper coloring.
Given these sets of vertices, the algorithm simply iterates the sets of vertices from $V_{r}$ to $V_1$ and color the current set of vertices given the coloring of vertices from previous sets.

To this end, each vertex $v \in V$ executes in parallel the following algorithm:\medskip \smallskip

\begin{algorithm}[H]\label{AlgPartition}\small
	\caption{partitions the vertices of the graph}
\SetInd{1em}{0em}
$\text{L}[v] \gets (\perp, \perp)$\\

	\For {\upshape $i \gets 1$ \textbf{to} $1 + 70\lceil \log_{2}n \rceil$} { \label{alg:outer}

	 \tcp{Step 1: removing removable cycles of length $4$ and degree $\le 4$}
            Collect the labeled neighborhood of $v$ up to distance $3$, and let $N(v)$ be the subgraph induced by all the vertices $u$ with $\text{L}[u] = (\perp, \perp)$ whose distance from $v$ is at most $3$. \\
             \If{\upshape $v$ belongs to some removable simple cycle $C$ in $N(v)$ of length $4$ and degree $\le 4$}
                    {$\text{L}[v] \gets (i,1)$ \\
                    $\text{key}[v] \gets \{\text{ID}(u) \mid u \in V(C)  \}$} \smallskip

            \tcp{Step 2: removing vertices of degree less than $4$}
            Collect the labeled neighborhood of $v$ up to distance $1$, and let $N(v)$ be the subgraph induced by all the vertices $u$ with $\text{L}[u] = (\perp, \perp)$ whose distance from $v$ is at most $1$. \\
            \If{\upshape $\text{L}[v] = (\perp, \perp)$ \textbf{and} $v$ is of degree less than 4 in $N(v)$} {\label{removeLowDeg} $\text{L}[v] \gets (i,2)$ \\$\text{key}[v] \gets \{\text{ID}(v)\}$
  }
    }
\label{alg:partitionToLevels}
\end{algorithm}\medskip \smallskip

We say that a vertex $v \in V$ is active at some point of the algorithm's execution if $\text{L}[v] = (\perp, \perp)$ at that point, or inactive otherwise. We claim that by the end of the algorithm all the vertices must be inactive. For this, we start by showing the following lemma:

\begin{lemma} For each $i \in \{1,...,1+70\lceil\log_{2}n\rceil\}$, if $A$ and $B$ respectively denote the sets of all the vertices that were active at the start of the outer loop's $i$-th iteration and at its end, then $|B| \le 0.99|A|$.\label{lem:numberOfVerticesDeactivates}\end{lemma}

\begin{proof} Let $L_A \subseteq A$ and $H_A \subseteq A$ be the sets of all vertices of degree less than $4$ and greater than $4$ in $G[A]$, respectively. Since each vertex who becomes inactive at some point of the algorithm's execution remains inactive until the end, it is enough to show that at least $0.01|A|$ of the vertices in $A$ become inactive by the end of the $i$-th iteration. This is clearly the case when $|L_A| \ge 0.01|A|$ as each vertex in $L_A$ who remained active after step 1, must become inactive by the end of step 2. 
So it remains to show that this is also the case when $|L_A| < 0.01|A|$. Let $U$ be the set of all vertices that remained active at the end of the first step of the $i$'th iteration. We divide the edges of the graph $G[U]$ into two disjoint subsets $E_{1}$ and $E_{2}$ where $E_{1}$ contains all the edges in $G[U]$ that are incident to some vertex in $H_A$, and $E_{2}$ contains all the other edges. Since $G[A]$ is a triangle-free planar graph, we have 
$|E(G[A])| \le 2|A|$ (by setting $g = 4$ in Claim \ref{claim:numberOfEdges}). It follows that $\Sigma_{v \in A}d(v, G[A]) = 2|E(G[A])| \le 4|A|$ and so $4|A| \ge \Sigma_{v \in A}d(v, G[A]) \ge \Sigma_{v \in A \setminus L_A} d(v, G[A]) = \Sigma_{v \in A \setminus L_A} (d(v, G[A]) - 4) + 4|A \setminus L_A|$. This implies that $4|L_A| = 4|A| - 4|A \setminus L_A| \ge \Sigma_{v \in A \setminus L_A} (d(v, G[A]) - 4) \ge |H_A|$, and so $\Sigma_{v \in H_A}d(v, G[A]) = \Sigma_{v \in H_A}(d(v, G[A]) - 4) + 4|H_A| \le 4|L_A| + 16|L_A| = 20|L_A| < 0.2|A|$. We conclude that $|E_{1}| \le \Sigma_{v \in H}d(v, G[A]) < 0.2|A|$.

Now, let $G'$ be the subgraph induced by the edges in $E_{2}$, that is the graph whose set of vertices $v$ is all vertices in $V$ such that $v$ is incident to at least one edge in $E_2$,
and whose set of edges is $E_2$. Since $G'$ is a subgraph of $G$ it must also be a triangle-free planar graph. Moreover, $G'$ cannot contain a simple cycle of length $4$. To see this, recall that $G'$ contains only vertices whose degree in $G[A]$ is at most $4$. So if $G[A]$ contains a cycle of at length at most 4 it must be a removable cycle (as $G[A]$ is a triangle free graph). However, the algorithm already removed in the first step of the $i$'th iteration, all vertices that are contained in removable cycle of length 4. It follows by Claim $\ref{claim:numberOfEdges}$ that $1.7|A| > \frac{5}{5-2}|V(G')|\ge |E(G')| = |E_{2}|$. It follows that $G[U]$ contains at most $|E_{1}| + |E_{2}| < 1.9|A|$ edges. Now, let $L_U$ be the set of all vertices of degree less than $4$ in $G[U]$. We have $3.8|A| \ge 2(|E_{1}| + |E_{2}|) = \Sigma_{v \in U}d(v, G[U]) \ge \Sigma_{v \in U \setminus L_U}d(v, G[U]) \ge 4(|U| - |L_U|) = 4(|A| - |A \setminus U| - |L_U|)$ and so $|A \setminus U| + |L_U| \ge 0.05|A|$. Note that in the first step of the $i$'th iteration, the algorithm inactivates $|A \setminus U|$ vertices (by definition of $U$) and in the second step it inactivates $|L_U|$ vertices.
Overall, the algorithm inactivates in the $i$'th iteration $|A \setminus U| + |L_U|$ vertices.
It follows that $|B| \le 0.95|A|$, as required. 
\end{proof}

We conclude the following lemma:

\begin{lemma} By the end of the algorithm all vertices are inactive. \label{lem:numberOfVerticesDeactivatesCor}\end{lemma}

\begin{proof}
By Lemma \ref{lem:numberOfVerticesDeactivates}, it easily follows that the number of vertices that remained active after $k$ iterations  is at most $0.99^{k}n$.
Note that the algorithm invokes $ 1 + 70\lceil\log_{2}n\rceil$ iterations therefore after the last iteration the number of vertices that remain active is $0.99^{1 + 70\lceil\log_{2}n\rceil}n < 1$. Hence, no vertex remain active by the end of the algorithm.\end{proof}

Let $\beta = 1 + 70\lceil \log_{2}n \rceil$ and let $S = \{1,...,\beta\} \times \{1,2\}$. The above algorithm partitions the vertices of the graph into $|S|$ disjoint subsets $\{H_{i,j}\}_{(i,j) \in S}$ where each $H_{i,j}$ contains all the vertices $u$ with $\text{L}[u] = (i,j)$.
Roughly speaking, our goal now is to go over the sets $H_{i,j}$ in a reverse order (that is, $H_{\beta,2}$, $H_{\beta,1}$, $H_{\beta-1,2}$, $H_{\beta-1,1}$,...,$H_{1,2}$, $H_{1,1}$) and color the vertices of the current set given the previous 4-coloring of all vertices belonging to previous sets.
We claim that the algorithm can always assign colors to the vertices in the current set in such a way that the coloring of all colored vertices so far (i.e., vertices from the current set and from previous sets)
forms a proper coloring.
To see this, we distinguish between two cases.
The first case is when the current set is of the form $H_{i,2}$ and the second case is when the current set is of the form  $H_{i,1}$ for some $i \in \{1, ..., \beta\}$.
Consider the first case, note that by construction the degree of all the vertices in $H_{i,2}$ in the induced graph of the set of vertices $\left(\cup_{\beta \geq j>i}{H_{j,1}}\right) \bigcup \left(\cup_{\beta \geq j \geq i}{H_{j,2}} \right)$
(the vertices from all previous and current sets) is less than 4.
Therefore, we can always find a free color for every vertex in $H_{i,2}$ from the palette $\{1,2,3,4\}$, and we can synchronize between the different vertices in this set by partitioning it into several disjoint independent sets.

Consider now the second case where the current set is of the form $H_{i,1}$.
Recall that all the vertices in $H_{i,1}$ belong to a removable cycle in the induced graph of the set of vertices $\left( \cup_{\beta \geq j\geq i}{H_{j,1}}\right) \bigcup \left( \cup_{\beta \geq j \geq i}{H_{j,2}} \right)$.
Our goal is to use Lemma \ref{lem:coloringRemovableCycle} to color all these vertices.
This by itself is not trivial as the cycles might overlap (by either sharing a common vertex or by being connected by an edge).
A first attempt to solve this issue is to partition the cycles into sets such that in each one of them there are no two overlapping cycles.
The problem with this approach is that when two (or more) cycles $C_1$ and $C_2$ overlap by sharing a vertex, we have to first color one of them, say $C_1$, but then when getting to color $C_2$ some of its vertices are already colored and so we cannot use Lemma \ref{lem:coloringRemovableCycle} to color $C_2$.
To overcome this, we loosely speaking do the following.
We look at the following super graph of removable cycles.
Each removable cycle is a node in the super graph and each two nodes in the super graph have an edge if their corresponding cycles overlap.
We first color the nodes in the super graph and then we partition the vertices of $H_{i,1}$ as follows.
Let $c_1,...,c_r$ be the set of colors used in the super graph (we will later show that $r$ is constant).
Each original vertex from $H$ picks the cycle with the maximal color in which it participates in. We order these sets as follows: the first set of vertices $H_{i,1}^1$ is the vertices of $H_{i,1}$ that picked a cycle of minimum color $c_1$, the second set $H_{i,1}^2$ are the vertices that picked a a cycle of color $c_2$ and so on.
We then color the vertices in this order, that is, first color the vertices in $H_{i,1}^1$ and then the vertices in $H_{i,1}^2$ and so on.
The key observation here is as follows.
Consider a vertex $v$ belonging to $H_{i,1}^j$ for some $1 \leq j \leq r$.
Let $C$ be the removable cycle $v$ participates in of the maximal color, note that by construction this color is $c_j$.
We claim that all vertices of $C$ belong to sets of the form $H_{i,1}^{j'}$ for some $j' \geq j$ and therefore when coloring the vertices
$H_{i,1}^j$ and $v$ in particular, we can always assign colors to vertices $H_{i,1}^j$.

We next define the super graph of cycles more formally.
Given a set of vertices $U$, we define $N(U)$ to be the set containing
 all the vertices $u$ whose distance in $G$ from some vertex with identifier in $U$ is at most $1$.
For each $i \in \{1, ..., \beta\}$, we denote by ${\cal G}_{i} = ({\cal V}_{i},{\cal E}_{i})$ the graph whose nodes and edges are $\{\text{key}[u] \mid u \in H_{i,1}\}$ and $\{\{\text{key}[u], \text{key}[v]\}\mid \{u,v\} \subseteq H_{i,1} \text{\; and \;} \text{key}[u] \ne \text{key}[v] \text{\; and \;} N(\text{key}[u]) \cap \text{key}[v] \ne \emptyset  \}$, respectively.

\begin{lemma}  For each $i \in \{1, ..., \beta\}$, the maximum degree of the graph ${\cal G}_{i}$ is less than $4^5$. \label{lem:degreeSuperGraph}\end{lemma}
\begin{proof} Let $i \in \{1, ..., \beta\}$ and $v \in H_{i,1}$. The degree of $\text{key}[v]$ in ${\cal G}_{i}$ is equal to the size of the set $\{\text{key}[u] \mid u \in H_{i,1} \text{\; and \;} N(\text{key}[u]) \cap \text{key}[v] \ne \emptyset  \text{\; and \;} \text{key}[u] \ne \text{key}[v]\}$.
In other words, the neighbors of $\text{key}[v]$ are nodes that represent cycles $C'$ in ${\cal G}_{i}$ that either have a common vertex with the cycle $C$ that $\text{key}[v]$ represents or there is an edge such that one endpoint of this edge is in $C$ and the other in $C'$.
It is not hard to verify that the size of this set is bounded by the size of $A' = \{u \in H_{i,1} \mid N(\text{key}[u]) \cap \text{key}[v] \ne \emptyset\}$.

Let $u \in A'$, let $C_u$ be the cycle that $\text{key}[u]$ represents and $C_v$ be the cycle that $\text{key}[v]$ represents.
By definition, we have $V(C_u) \cap V(C_v) \neq \emptyset$ or there is an edge $\{x,y\} \in E$ such that $x \in V(C_u)$ and $y \in V(C_v)$.
Since both cycles  $C_u$ and $C_v$ are of length 4, it is not hard to verify that the distance between $u$ and $v$ in the induced graph of
$V(C_u) \cup V(C_v)$ is at most 5.
Moreover, all vertices of $V(C_u) \cup V(C_v)$ belong to $H_{i,1}$.
We get that the distance between $u$ and $v$ in the induced graph $G[H_{i,1}]$ is at most $5$.
Note that the maximum degree in $G[H_{i,1}]$ is at most 4.
Straight forward calculations show that for a vertex $v'$ in a graph of degree at most $4$ there could at most $4^5$ vertices at distance at most $5$ from it (we didn't try optimize constants). Hence there could be at most $4^5$ such vertices $u$ in the set $A'$.
%
\end{proof}

For each $i \in \{1, ..., \beta\}$, we want to compute a proper $5$-coloring $\varphi_{i}$ of the graph $G[H_{i,2}]$ so that each vertex $v$ with $\text{L}[v] = (i,2)$ knows the value of $\varphi[v] = \varphi_{i}(v)$. This can be done in $O(\log^{*}n)$ by executing in parallel a $\Delta+1$-coloring algorithm (e.g. \cite{Goldberg87}) on the graph $G[H_{i,2}]$. Next, we want to compute a proper coloring $\varphi_{i}:{\cal V}_{i} \to \{1,2,...,4^5\}$ of the graph ${\cal G}_{i}$ so that each vertex $v$ with $\text{L}[v] = (i,1)$ knows the value of $\varphi[v] = \varphi_{i}(\text{key}[v])$. This can be done in $O(\log^{*}n)$ rounds by simulating in parallel a $4^5$-coloring algorithm for each ${\cal G}_{i}$ on the graph $G$.\medskip

The simulation can be done as follows.
First we show that we can assign each node in ${\cal G}_{i}$ an integer identifier that is polynomial in $n$.
If ${\cal G}_{i}$ was the underling graph, then it is possible to map each $\text{key}$ to an ID in the range $\{1,...,n^{4c'}\}$ (as the ID of each vertex is in the range $\{1,...,n^{c'}\}$ and each $\text{key}$ consists of 4 vertices in the original graph).
This way, we get a graph of maximum degree $< 4^5$ (by Lemma \ref{lem:degreeSuperGraph}) and maximum ID $\le n^{4c'}$, and so we can find a proper $4^5$-coloring in $O(\log^*{n})$. But, the underling graph is $G$ and so we have to simulate this algorithm by forwarding each message of this algorithm to all the corresponding vertices in the graph $G$ (for example, each node $k$  in ${\cal G}_{i}$ can be simulated by the vertex with the highest identifier in the cycle that $k$  represents). This increases the running time by a constant factor (as the length of each cycle in ${\cal G}_{i}$ is of constant length).

%
%

The following algorithms are executed by each vertex $v$.
Algorithm \ref{AlgColorSuperGraph} assigns a color to every $v$ as explained above.
Note that these are not the final colors of the vertices but rather colors that are supposed to be used
to partition the vertices in each $H_{i,j}$ into a small number of sets such that each such set can be colored simultaneously in parallel.
Algorithm \ref{AlgColorGraph} then assigns each vertex the final color.  \medskip \smallskip

\begin{algorithm}[H]\label{AlgColorSuperGraph}\small
	\caption{Give a color $\varphi[v]$ for every vertex $v$. The colors $\varphi[v]$ are supposed to synchronize between the different vertices in each $H_{i,j}$ for $i \in \{1, ..., \beta\}$ and $1 \leq j \leq 2$}
\SetInd{1em}{0em}
Let $i,j$ be the the indices such that $L(v) = (i,j)$.\\
If $j=2$ then execute the $\Delta+1$-coloring algorithm of \cite{Goldberg87} in the graph $G[H_{i,j}]$ and set $\varphi[v]$ to be the color assigned to $v$.\\
If $j=1$  then execute the $\Delta+1$-coloring algorithm of \cite{Goldberg87} in the super graph ${\cal G}_i$ (by simulating the graph ${\cal G}_i$ such that every cycle $\text{key}[u]$ is simulated by the vertex with highest ID in $\text{key}[u]$) and set $\varphi[v]$ to be the color assigned to $\text{key}[v]$.
\label{alg2:col}
 \end{algorithm}\medskip \smallskip

Now, in order to find a $4$-coloring of the whole graph each vertex $v \in V$ executes in parallel the following algorithm: \medskip \smallskip


\begin{algorithm}[H]\label{AlgColorGraph}\small
	\caption{Color the vertices of the graph}
\SetInd{1em}{0em}
$\text{C}[v] \gets \;\perp$\\
Collect the value of $(\text{key}[u], \varphi[u])$ from each vertex $u$ with $\text{L}[v] = \text{L}[u]$ whose distance from $v$ is at most $2$, and let $\text{P}[v]$ be the set containing all the pairs (among the collected ones) whose key contains the ID of $v$.\\ \label{line:collectCycle}
Choose $(\text{k},\text{color}) \in \text{P}[v]$ with maximum color, and set $\text{key}_{\text{new}}[v] \gets \text{k}$ and $\varphi_{\text{new}}[v] \gets \text{color}$.\label{line:chooseCycle}

	\For {\upshape $i \gets 1 + 70\lceil \log_{2}n \rceil$ \textbf{downto} $1$} { \label{line:outer}

\tcp{Step 1: coloring all vertices $v \in V$ with $\text{L}[v] = (i,2)$}
		    \For {\upshape $k \gets 1$ \textbf{to} $5$} { \label{line:innerStep1}
		Collect the value of $\text{C}[u]$ from each neighbor $u$ of $v$, and let $\text{Colors}$ be the set containing all these values.

             \If{\upshape $\text{L}[v] = (i,2)$ \textbf{and} $\varphi_{\text{new}}[v] = k$}  { \label{line:innerCond1}
		Choose a color $x \in \{1,2,3,4\} \setminus \text{Colors}$, and set $\text{C}[v] \gets x$.\label{line:ColorVertex}
      }}

                   \smallskip\tcp{Step 2: coloring all vertices $v \in V$ with $\text{L}[v] = (i,1)$}
		    \For {\upshape $k \gets 1$ \textbf{to} $4^5$} { \label{alg2:inner}
		Collect the labeled neighborhood of $v$ up to distance $3$ (along with the value of $C$ of each vertex in that neighborhood), and let $N_{i,k}(v)$ be the subgraph induced by all the vertices whose ID belongs to $N(\text{key}_{\text{new}}[v])$.

             \If{\upshape $\text{L}[v] = (i,1)$ \textbf{and} $\varphi_{\text{new}}[v] = k$}  {\label{line:innerCond2}
Find a proper $4$-coloring of $N_{i,k}(v)$ which is consistent with the vertices that are already colored (by using a deterministic algorithm), and set $\text{C}[v]$ accordingly.  \label{line:ColorRemovable}

}}}

\label{alg:finalColor}
\end{algorithm}\medskip \smallskip

We say that a set of vertices $U \subseteq V$ is properly colored at some round $r$, if in the start of this round we have $\text{C}[v] \ne \; \perp$ for every vertex $v \in U$, and $\text{C}[v] \ne \text{C}[u]$ for every two adjacent vertices $u$ and $v$ in $G[U]$. We also say that a set of vertices $U \subseteq V$ is not colored at some round $r$, if in the start of this round we have $\text{C}[v] = \; \perp$ for every vertex $v \in U$. For each vertex $v \in V$, we assume here that $\varphi_{\text{new}}[v]$ and $\text{key}_{\text{new}}[v]$ refer to the values that $v$ obtained after line \ref{line:chooseCycle}. For each $i \in \{0, ..., \beta\}$, we let $A_i$ denote the set of all vertices $u \in V$ with $\text{L}[u] \in \{j \mid i < j \le \beta\} \times \{1,2\}$.


The next lemma shows that when a vertex wants to pick a color in line \ref{line:ColorVertex} then a free color is always available, that is, it can pick a color from $\{1,...,4 \}$ that is not yet used by any of its neighbors in $G$, and that the resulting coloring by the end of step $1$ is proper.

\begin{lemma}  For all $i \in \{1, ..., \beta\}$, if at the start of the $i$-th iteration of the loop at line \ref{line:outer}, we have that $A_i$ is properly colored and $V \setminus A_i$ is not colored, then by the end of step $1$ of the $i$-th iteration of this same loop, we get that $A_i \cup H_{i,2}$ is properly colored and $V \setminus (A_i \cup H_{i,2})$ is not colored.\end{lemma}

\begin{proof} Let $i \in \{1, ..., \beta\}$. Assume that in the start of the $i$-th iteration of the loop at line \ref{line:outer}, we have that $A_i$ is properly colored and that $V \setminus A_i$ is not colored. Let $u \in H_{i,2}$. There must be some iteration of the loop at line \ref{line:innerStep1} in which $u$ detects that the condition at line \ref{line:innerCond1} holds. As the set $A_i \cup H_{i,2}$ is equal to the set of all vertices that were active at the moment when $u$ became inactive in algorithm $\ref{AlgPartition}$, we get that the degree of $u$ in the graph $G[A_i \cup H_{i,2}]$ is less than $4$, and so it can be adjacent to at most $3$ colored neighbors. It follows that $u$ can always find a color at line \ref{line:ColorVertex} which is different from all of its colored neighbors. It is left to show that by the end of step $1$, we have $\text{C}[u] \ne \text{C}[v]$ for every two adjacent vertices $u$ and $v$ in $G[H_{i,2}]$. Let $u$ and $v$ be two adjacent vertices in $G[H_{i,2}]$. Since they are adjacent, we must have that $\varphi_{\text{new}}[v] \ne \varphi_{\text{new}}[u]$ and so we can assume with out loss of generality that $\varphi_{\text{new}}[v] > \varphi_{\text{new}}[u]$. This means that $u$ is colored before $v$ and so when $v$ is colored it must be with a color different from $u$. \end{proof}

Next, we show that it is always possible to find the required coloring at line \ref{line:ColorRemovable}. We start with the following lemma:


\begin{lemma} For any $(i,k) \in \{1, ..., \beta\} \times \{1, ..., 4^5\}$ and $v \in V$, if $\text{L}[v] = (i,1)$ and $\varphi_{\text{new}}[v] = k$ then at the start of the $k$-th iteration of the loop in line \ref{alg2:inner} during the $i$-th iteration of the loop in line \ref{line:outer} we must have $\text{C}[u] = \; \perp$ for every vertex whose ID belongs to $\text{key}_{\text{new}}[v]$.
\label{lem:uncoloredCycle}\end{lemma}

\begin{proof} Let $(i,k) \in \{1, ..., \beta\} \times \{1, ..., 4^5\}$ and $v \in V$ with $\text{L}[v] = (i,1)$ and $\varphi_{\text{new}}[v] = k$. Let $F$ be the set of all vertices whose ID belongs to $\text{key}_{\text{new}}[v]$, and for each such vertex $u \in F$, denote by $\text{C}_{i,k}[u]$ the value of $\text{C}[u]$ at the start of the $k$-th iteration of the loop in line \ref{alg2:inner} during the $i$-th iteration of the loop in line \ref{line:outer}. Assume towards a contradiction that there is a vertex $u \in F$ with $\text{C}_{i,k}[u] \ne \;\perp$.
We claim that $F \subseteq H_{i,1}$. To see this, note that as $\text{L}[v] = (i,1)$ then $\text{key}_{\text{new}}[v]$ is a removable cycle that was picked by some vertex $u'$
in the $i$'th iteration of Algorithm \ref{AlgPartition}.
This means, that all the vertices whose ID in $\text{key}_{\text{new}}[v]$ were active in the $i$'th iteration of Algorithm \ref{AlgPartition}.
Moreover, all these vertices belonged to a removable cycle in the $i$'th iteration of Algorithm \ref{AlgPartition} and therefore by construction $\text{L}[u'] = (i,1)$ for all $u' \in F$, that is,  $F \subseteq H_{i,1}$.
Since we assume that $\text{C}_{i,k}[u] \ne \;\perp$, it must be that $\varphi_{\text{new}}[u] < k$.
This means that the set $\text{P}[u]$ did not contain the pair $(\text{key}_{\text{new}}[v],\varphi_{\text{new}}[v])$, as otherwise $u$ wouldn't have picked
$(\text{key}_{\text{new}}[u], \varphi_{\text{new}}[u])$ as $\varphi_{\text{new}}[u] < \varphi_{\text{new}}[v] = k$.

Since $\text{key}_{\text{new}}[v]$ contains $\text{ID}(u)$ and $\text{L}[u] = \text{L}[v]$, it must be that no vertex in $H_{i,1}$ whose distance from $u$ is at most $2$ had its key equal to $\text{key}_{\text{new}}[v]$ at line \ref{line:collectCycle} as $u$ would have collected it.

However, we claim that $v$ could have got $\text{key}_{\text{new}}[v]$ only from a vertex $x$
whose distance from $u$ is at most $2$, which is a contradiction.
To see this, note that $\text{key}_{\text{new}}[v]$ must have been picked as $\text{key}_{\text{new}}[x]$ by a vertex $x$ whose ID belongs to
$\text{key}_{\text{new}}[v]$ in Algorithm \ref{AlgPartition}.
In other words, both $x$ and $u$ belong to the removable cycle that $\text{key}_{\text{new}}[v]$ represents. As the removable cycle is of length 4, then the distance from $u$ to $x$ is at most $2$, contradiction.\end{proof}

Our next goal is to show that by the end of the algorithm all vertices are colored and more over the coloring of all vertices is a proper coloring.

\begin{lemma}  For all $i \in \{1, ..., \beta\}$, if at the start of the second step of the $i$-th iteration of the loop at line \ref{line:outer}, we have that $A_i \cup H_{i,2}$ is properly colored and $V \setminus (A_i \cup H_{i,2})$ is not colored, then by the end of step $2$ of the $i$-th iteration of this same loop, we get that $A_i \cup H_{i,2} \cup H_{i,1}$ is properly colored and $V \setminus (A_i \cup H_{i,2} \cup H_{i,1})$ is not colored.\end{lemma}

\begin{proof} Let $i \in \{1, ..., \beta\}$ and $u \in H_{i,1}$. There must be some iteration of the loop at line \ref{alg2:inner} in which $u$ detects that the condition at line \ref{line:innerCond2} holds. As the set $A_i \cup H_{i,2} \cup H_{i,1}$ is equal to the set of all vertices that were active at the moment when the vertex $u$ became inactive, we get that the vertices whose ID in $\text{key}_{\text{new}}$ induce a removable cycle of degree $\le 4$ in $G[A_i \cup H_{i,2} \cup H_{i,1}]$. Since by Lemma \ref{lem:uncoloredCycle} this cycle is not colored, we get from Lemma \ref{lem:coloringRemovableCycle} that it is always possible to find a proper $4$-coloring of $N(\text{key}_{\text{new}}[u])$ which is consistent with the vertices that are already colored. It is left to show that by the end of step $2$, we have $\text{C}[u] \ne \text{C}[v]$ for every two adjacent vertices $u$ and $v$ in $G[H_{i,1}]$. Let $u$ and $v$ be two such vertices. We distinguish between two cases: $\varphi_{\text{new}}[v] = \varphi_{\text{new}}[u]$ and $\varphi_{\text{new}}[v] \ne \varphi_{\text{new}}[u]$. We claim that in the first case we must have $\text{key}_{\text{new}}[u] = \text{key}_{\text{new}}[v]$. Indeed, assume towards a contradiction that $\text{key}_{\text{new}}[u] \ne \text{key}_{\text{new}}[v]$. Since $\varphi_{\text{new}}[u] = \varphi_{\text{new}}[v]$, we must have $N(\text{key}_{\text{new}}[u]) \cap \text{key}_{\text{new}}[v] = \emptyset$ which is impossible as $u$ and $v$ are adjacent in $G$. It follows that $N_{i,k}(v) = N_{i,k}(u)$ and since each vertex encodes this graph in the same way, they must compute the same $4$-coloring of it (as they are using the same deterministic algorithm) and so their colors will be different. In the second case, we can assume without loss of generality that $\varphi_{\text{new}}[v] > \varphi_{\text{new}}[u]$. This means that $u$ is colored before $v$ and so when $v$ is colored it must be with a color which is different from the color of $u$.\end{proof}

By definition, we have $A_{i-1} = A_i \cup H_{i,1} \cup H_{i,2}$ for all $i \in \{1,...,\beta\}$, and so we can conclude the following:

\begin{corollary}  For all $i \in \{1, ..., \beta\}$, if $A_{i}$ is properly colored at the start of the $i$-th iteration of the loop at line \ref{line:outer}, then $A_{i-1}$ is properly colored at the end of the $i$-th iteration of this loop.\end{corollary}

Since $A_{\beta} = \emptyset$ and $A_{0} = V$ we get that $V$ is properly $4$-colored by the end of the algorithm. The following simple lemma shows that the running time of the algorithm is $O(\log{n})$.

\begin{lemma} The total number of communication rounds of the algorithm is $O(\log{n})$. \label{lem:complexity}\end{lemma}
\begin{proof} Each iteration of the outer loop in algorithm \ref{AlgPartition} and \ref{AlgColorGraph} requires a constant number of communication rounds and so each of them requires $O(\log{n})$ communication rounds in total. Algorithm \ref{AlgColorSuperGraph} requires $O(\log^*{n})$ communication rounds. \end{proof}


\section{Finding a 6-coloring of planar graphs}
\label{sec:6-coloring}

The general framework of the 6-coloring algorithm is similar to the algorithm given in the previous section.
Similarly, the algorithm consists of three parts.
The first part, given in Algorithm \ref{Alg6ColorPartition}, partitions the vertices into $O(\log{n})$ disjoint sets $H_{i,j}$ for each $(i,j) \in \{1, ..., \beta\} \times \{1,2\}$ (for some $\beta= O(\log{n})$ to be fixed later on). This is done by iteratively removing vertices that are part of some removable cycle or that are of degree less than $6$.
The second part, given in Algorithm \ref{Alg6ColorSuperGraph}, assigns each vertex $v$ a color $\varphi[v]$. Similarly to the previous section this colors are not the final colors but are rather used to synchronize between the different vertices in each $H_{i,j}$.
Finally, Algorithm \ref{Alg6ColorGraph} assigns each vertex its final color. \smallskip

We start by giving some technical lemmas that will be used in our analysis later on and then give the formal pseudo code of Algorithms \ref{Alg6ColorPartition}, \ref{Alg6ColorSuperGraph} and \ref{Alg6ColorGraph}. The next auxiliary lemma will be crucial in our analysis:

\begin{lemma} For any graph $G$ and any two different, not edge-disjoint, simple cycles $C_{1}$ and $C_{2}$ in $G$, if $G$ does not contain any removable cycle of length at most $|C_{1}| + |C_{2}|$, then the graph $G[V(C_{1}) \cup V(C_{2})]$ is isomorphic to a complete graph. \label{lem:edgeDisCycles}\end{lemma}

The proof of the lemma is a quite technical. We first prove the following claim:

\begin{claim} For any graph $G$ and any two different, not edge-disjoint, simple cycles $C_{1}$ and $C_{2}$ in $G$, if $G$ does not contain any removable cycle of length at most $|C_{1}| + |C_{2}|$, then both $G[V(C_{1})]$ and $G[V(C_{2})]$ are isomorphic to a complete graph. \label{claim:EachCycleIsClique} \end{claim}

\begin{proof} Let $C_{1}$ and $C_{2}$ be two different, not edge-disjoint, simple cycles in a graph $G$, and assume that $G$ does not contain any removable cycle of length at most $|C_{1}| + |C_{2}|$. For each $i \in \{1,2\}$, let $H_{i} = (V(C_{i}), E(C_{i}))$ be the cycle graph $C_{i}$.

For convenience, we only prove the claim for $G[V(C_{1})]$ (as the proof for $G[V(C_{2})]$ is symmetric). Clearly, if $|C_{1}|$ is even, then $G[V(C_{1})]$ is isomorphic to a complete graph as $G$ does not contain any removable cycle of length $|C_{1}|$. So it is left to show that $G[V(C_{1})]$ is isomorphic to a complete graph when $|C_{1}|$ is odd. To this end, we consider separately the cases $V(C_{2}) \subseteq V(C_{1})$ and $V(C_{2}) \not\subseteq V(C_{1})$.

If $V(C_{2}) \subseteq V(C_{1})$, then every edge in $E(C_{2})$ must also be in $G[V(C_{1})]$. This means that $G[V(C_{1})]$ contains two different cycles $C_{1}$ and $C_{2}$, and so it cannot be isomorphic to a cycle graph. It follows that $G[V(C_{1})]$ must be isomorphic to a complete graph as $G$ does not contain any removable cycle of length $|C_{1}|$.

If $V(C_{2}) \not\subseteq V(C_{1})$, then there exist a vertex $v \in V(C_{2}) \setminus V(C_{1})$.
Let $P = (V_{P},E_{P})$ be the connected component of $v$ in $H_{2}[V(C_{2}) \setminus V(C_{1})]$. Clearly, $V(C_{1}) \cap V(C_{2}) \ne \varnothing$ (as $C_{1}$ and $C_{2}$ are not edge-disjoint and so also not vertex-disjoint) which implies that $P$ is a connected proper subgraph of $H_{2}$ (that is, $P$ is a subgraph of $H_2$ such that $P \neq H_2$).
Now, $P$ is a connected proper subgraph of a cycle graph, and so it must be isomorphic to a path graph.
Let $k$ be the number of vertices in $V_P$. Note that $1 \le k < |C_{2}|$.
Let $p = (v_{1},...,v_{k})$ be the simple path that $P$ is isomorphic to.

We say that a vertex $u \in V(C_{2}) \cap V(C_{1})$ is an intersection point of $p$ with $C_{1}$ if either $\{u,v_{1}\} \in E(C_{2})$ or $\{u,v_{k}\} \in E(C_{2})$. We claim that there exist two different intersection points of $p$ with $C_{1}$.
Indeed, if $k = 1$, then $v_{1}$ has two different neighbors $x$ and $y$ in $H_{2}$ (since $|C_{2}| > 2$), and these neighbors also belong to $V(C_{1})$ as $P$ is the maximal connected component of $v$ in $H_{2}[V(C_{2}) \setminus V(C_{1})]$.
Otherwise, $k > 1$, and so $v_{1} \ne v_{k}$. Since $H_{2}$ is a cycle graph with $> 2$ vertices, each one of the vertices in $\{v_{1}, v_{k}\}$ has a neighbor in $H_{2}[V(C_{2}) \setminus V_{p}]$, and this neighbor also belongs to $V(C_{1})$ as $P$ is the maximal connected component of $v$ in $H_{2}[V(C_{2}) \setminus V(C_{1})]$.
Denote these neighbors by $x$ and $y$ (where $y$ is the neighbor of $v_{1}$ and $x$ is the neighbor of $v_{k}$), and assume towards a contradiction that $x = y$. The sequence $\tilde{C} = (v_{1},...v_{k},x)$ is a simple cycle of length $> 2$ in the graph $H_{2}$, and so it must be that $H_{2} = (V(\tilde{C}),E(\tilde{C}))$ which means that $|V(C_{1}) \cap V(C_{2})| = 1$ and so $C_{1}$ and $C_{2}$ must be edge-disjoint which is, of course, a contradiction (since assume that $C_1$ and $C_2$ are not edge disjoint) so $x \ne y$.
\begin{figure}[!ht]
	\centering
	\includegraphics[trim={0 0.3cm 12.9cm 0.1cm},clip]{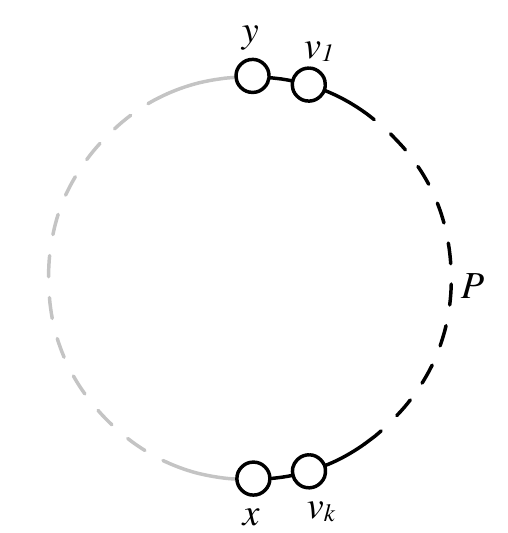}
	\caption{An illustration of the cycle graph $H_2$. The path $P$ is the maximal connected of $v$ in the graph $H_2[V(C_2) \setminus V(C_1)]$ where the vertices $v_1$ and $v_k$ are its endpoints (possibly $v_1 = v_k$). The vertices $x$ and $y$ are the intersection points of $p$ with $C_1$.}
\end{figure}

Now, since $H_{1}$ is a cycle graph with $> 2$ vertices, there exist in $H_{1}$ two simple, edge-disjoint paths $p_{1}$ and $p_{2}$ such that each of them goes from $x$ to $y$ and has $>1$ vertices.
Since $|C_{1}|$ is odd, exactly one of the paths $p_{1}$ and $p_{2}$ is of odd length and the other is of even length.
Therefore, exactly one of the sequences $p \circ p_{1}$ and $p \circ p_{2}$ is a simple cycle of even length greater than $2$
(note that $V(p_1) \cap V(p) = \emptyset$ and $V(p_2) \cap V(p) = \emptyset$ therefore both $p \circ p_{1}$ and $p \circ p_{2}$ are simple cycles).
\begin{figure}[!ht]
	\centering
	\includegraphics[trim={6cm 0.3cm 6.3cm 0.1cm},clip]{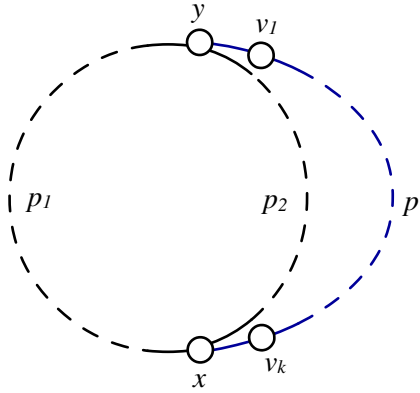}
	\caption{The cycle graph $H_1$ is colored in black. The simple cycle $C_1$ is decomposed into two different edge-disjoint paths $p_1$ and $p_2$. The simple cycle $C'$ is chosen to be either $p \circ p_1$ or $p \circ p_2$ (the one with even length).}
\end{figure}
Denote this cycle by $C'$. Since the length of $C'$ is at most $|C_{1}| + |C_{2}|$, it must be that the graph $G[V(C')]$ is isomorphic to a complete graph (again as $G$ does not contain any removable cycle of length at most $|C_{1}| + |C_{2}|$).

It is easy to see that $x$ has exactly one neighbor in $C'$ which also belongs to $V_{p}$. This vertex is, of course, $v_k$. Let $u_{2}$ be the other neighbor of $x$ in $C'$. Since $G[V(C')]$ is isomorphic to a complete graph there is an edge between $v_k$ and $u_{2}$.
Now, let $q$ be the longer path from $x$ to $u_{2}$ in the graph $H_{1}$ (note that both $x$ and $u_2$ belong to $H_1$ and, moreover, they are connected by an edge in $H_1$ so the longer path is the path from $x$ to $u_2$ that is not the direct edge from $x$ to $u_2$).

We claim that the sequence $C'' = (v_k) \circ q$ is a simple cycle of length $|C_{1}| + 1$. To see this, note first that there is an edge between $v_k$ and $x$ and an edge between $u_2$ and $v_k$ (as proven above).
\begin{figure}[!ht]
	\centering
	\includegraphics[trim={12.5cm 0.1cm 0.1cm 0.1cm},clip]{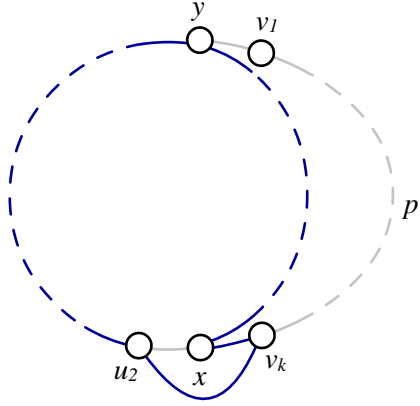}
	\caption{The vertex $x$ has two neighbors in $C'$ which are $v_k$ and $u_2$ (where $u_2$ must belong to $V(C_1)$ as $x$ has in $C'$ at least one neighbor in $V(C_1)$ and $v_k \not\in V(C_1)$). The blue path that goes from $x$ to $y$ and then continues to $u_2$ (where it might be that $u_2 = y$) is the path $q$, that is, the longer path from $x$ to $u_2$ in $H_1$. The simple cycle $C'' = (v_k) \circ q$ is colored in blue.} \end{figure}
The sequence $C'' = (v_k) \circ q$ is composed of the edge $\{v_k,x\}$, the path $q$ and then the edge $\{u_2,v_k\}$. We get that $C''$ is a simple cycle of length $|q| + 2$. Note also that $|q| = |C_1| -1$ (as the path $q$ consists of the entire cycle $C_1$ but the edge $\{x,u_2\}$). As $|C_1|$ is odd we get that $|C''|$ is even, and we also have $|C''| \le |C_{1}| + |C_{2}|$. It follows that $G[V(C'')]$ must be isomorphic to a complete graph and so also $G[V(C_{1})]$ (as $V(C_{1}) \subseteq V(C'')$). \end{proof}

\begin{proofof}{Lemma~\ref{lem:edgeDisCycles}} Let $u$ and $v$ be two different vertices from $V(C_{1}) \cup V(C_{2})$. If there is a cycle (either $C_1$ or $C_2$) to which both $u$ and $v$ belong, then by Claim \ref{claim:EachCycleIsClique}, $\{u,v\} \in E$. Otherwise, we can assume, without loss of generality, that $u \in V(C_{1}) \setminus V(C_{2})$ and $v \in V(C_{2}) \setminus V(C_{1})$.

Since $C_1$ and $C_2$ are not edge-disjoint, we must have $|V(C_1) \cap V(C_2)|\ge 2$. In other words, there are two different vertices $x$ and $y$ in $V(C_1) \cap V(C_2)$.
Moreover, note that $x \neq u$ and $y \neq u$, as $u \notin V(C_{2})$ and $x,y \in V(C_{2})$.
Similarly, $x \neq v$ and $y \neq v$.
We get that $u,v,x$ and $y$ are four different vertices.
The vertices $x,y$ and $u$ satisfy $x,y,u \in V(C_1)$, therefore by Claim \ref{claim:EachCycleIsClique} the graph $G$ contains the edges $\{u,y\}$ and $\{u,x\}$.
Similarly, we can show that the graph $G$ contains the edges $\{v,y\}$ and $\{v,x\}$. It follows that $(u,x,v,y)$ is a simple cycle in $G$ whose length is $4$ which is at most $2 + 2 < |C_{1}| + |C_{2}|$ and so $G[\{u,x,v,y\}]$ must be isomorphic to a complete graph which means that $\{u,v\} \in E$ as $G$ does not contain removable cycles of length at most $|C_{1}| + |C_{2}|$.
%
\end{proofof}

\medskip

From this lemma it easy to conclude the following lemmas:

\begin{lemma} For any planar graph $H$ and any two different simple cycles $C_1$ and $C_2$ in $H$ such that one of them is of length $5$ and the other is of length at most 5, if $H$ does not contain any removable cycle of length at most $10$ then $C_1$ and $C_2$ are edge-disjoint. \label{lem:cycleIsDisjoint}\end{lemma}

\begin{proof} Let $H$ be a planar graph that does not contain any removable cycle of length at most $10$, and let $C_1$ and $C_2$ be two different simple cycles in $H$ such that one of them is of length $5$ and the other is of length at most $5$. Assume towards a contradiction that $C_1$ and $C_2$ are not edge-disjoint. It follows from Lemma \ref{lem:edgeDisCycles} that $H[V(C_1) \cup V(C_2)]$ must be isomorphic to a complete graph, and so $H$ contains a $5$-clique as a subgraph which is a contradiction to the fact that $H$ is planar.\end{proof}

\begin{lemma} For any planar graph $H$ and any $4$-clique $K$ in $H$, if $H$ does not contain any removable cycle of length at most $10$ then
\begin{enumerate}
	\item \label{lem:case1} Every simple cycle $C$ in $H$ of length $5$ does not share any edge with $K$.
	\item \label{lem:case2} Every $4$-clique $K'$ in $H$ does not share any edge with $K$ unless $K = K'$.
	\item \label{lem:case3} Every simple cycle $C$ in $H$ of length $3$ does not share any edge with $K$ unless $V(C) \subseteq V(K)$.
\end{enumerate} \label{lem:cliqueIsDisjoint}\end{lemma}

\begin{proof} Let $H$ be a planar graph that does not contain any removable cycle of length at most $10$ and let $K$ be some $4$-clique in $H$.
\begin{enumerate}
	\item Let $C$ be some simple cycle in $H$ of length $5$, and assume towards a contradiction that $C$ and $K$ are not edge-disjoint, that is, assume that there exist an edge $\{x,y\} \in E(C) \cap E(K)$. Let $\{a,b\} = V(K) \setminus \{x,y\}$. Since $K$ is a $4$-clique in $H$, it must be that $C' = (x,y,a,b)$ is a simple cycle in $H$ which has at least one common edge with $C$ (the edge $\{x,y\}$), but by Lemma \ref{lem:cycleIsDisjoint} the cycles $C'$ and $C$ must be edge-disjoint, a contradiction.
	\item Let $K'$ be a $4$-clique in $H$ such that $K' \ne K$. Since $K$ and $K'$ are different cliques, we must have $V(K) \ne V(K')$ and so $|V(K) \cup V(K')| \ge 5$. Assume towards a contradiction that $K$ and $K'$ are not edge-disjoint. This means that there is an edge $\{x,y\} \in E(K) \cap E(K')$. Let $\{a,b\} = V(K) \setminus \{x,y\}$ and let $\{c,d\} = V(K') \setminus \{x,y\}$. The sequences $C_{1} = (a,b,x,y)$ and $C_{2} = (c,d,x,y)$ are two different, not edge-disjoint, simple cycles in $H$. Since $H$ does not contain any removable simple cycle of length at most $10$, it follows from Lemma \ref{lem:edgeDisCycles} that $H[V(C_{1}) \cup V(C_{2})]$ is isomorphic to a complete graph, and so $H$ contains a $5$-clique as a subgraph, in contradiction to the fact that $H$ is planar.
	\item Let $C$ be some simple cycle in $H$ of length $3$ such that $V(C) \not\subseteq V(K)$. Again, assume towards a contradiction that $C$ and $K$ are not edge-disjoint. Let $\{x,y\} \in E(C) \cap E(K)$ and $\{a,b\} = V(K) \setminus \{x,y\}$. Since $K$ is a $4$-clique in $H$, it must be that $C' = (x,y,a,b)$ is a simple cycle in $H$ which has at least one common edge with $C$. It follows from Lemma \ref{lem:edgeDisCycles} that $H[V(C') \cup V(C)]$ is isomorphic to a complete graph, and since $V(C) \not\subseteq V(K)$, we must have $|V(C') \cup V(C)|=|V(K) \cup V(C)| \ge 5$ and so again $H$ contains a $5$-clique as a subgraph which is a contradiction.\qedhere\end{enumerate}\end{proof}

\begin{lemma} It is possible to remove $3$ edges from the complete graph on $4$ vertices so that the resulting graph does not contain any simple cycle. \label{lem:removingEdges}\end{lemma}

\begin{proof} Let $H = (V,E)$ be the complete graph on $4$ vertices. Assume that $V = \{x,y,z,w\}$, and remove from $E$ the edges $\{x,y\}$, $\{y,z\}$ and $\{z,w\}$. The remaining edges are $\{x,z\}$, $\{y,w\}$ and $\{x,w\}$ which means that the resulting graph is the path $(z,x,w,y)$ and so it does not contain any simple cycle. \end{proof}

The next lemma will be used in our analysis to show that after we remove from the graph all removable cycles of length at most 10, the resulting graph does not contain too many edges.

\begin{lemma} For any planar graph $H$ with maximum degree at most $6$, if $H$ does not contain any removable cycle of length at most $10$, then $|E(H)| \le 2.9|V(H)|$.\label{lem:numberOfEdgesWithoutRemovableCycles}\end{lemma}

\begin{proof} Let $H$ be a planar graph that satisfies the conditions of the lemma (that is, the maximum degree of $H$ is at most $6$, and $H$ does not contain any removable cycle of length at most $10$). We start by showing that it is always possible to remove from $E(H)$ at most $1.5|V(H)|$ edges such that the resulting graph does not contain any simple cycle of length at most $5$. To this end, we are going to describe a procedure that will iteratively remove edges from $E(H)$ until no simple cycle of length at most $5$ remains. For the sake of analyzing the number of removed edges, this procedure is going to maintain for each vertex in $H$ some non-negative charge, and will guarantee, by updating these charges as edges are being removed, that the total charge of the vertices in the graph is always equal to the total number of removed edges. This way, we can get a bound on the total number of removed edges by bounding the total charge of the vertices.
Consider applying on $H$ the following procedure:
\begin{enumerate}
\item Give each vertex in $V(H)$ an initial charge of $0$.
	\item \label{proc:step2} As long as $H$ contains a simple cycle $C$ of length $5$, choose one arbitrary edge from $E(C)$, remove it from $E(H)$ and add $1/5$ to the charge of each vertex in $V(C)$.
	
	\item \label{proc:step3} As long as $H$ contains a $4$-clique $K$ as a subgraph, choose three edges from $E(K)$ as in Lemma \ref{lem:removingEdges}, remove them from $E(H)$ and add $3/4$ to the charge of each vertex in $V(K)$.
	
	\item \label{proc:step4} As long as $H$ contains a simple cycle $C$ of length $3$, choose one arbitrary edge from $E(C)$, remove it from $E(H)$ and add $1/3$ to the charge of each vertex in $V(C)$.
\end{enumerate}

\noindent We prove the following bound on the charge of each vertex:

\begin{claim} By the end of the procedure, the charge of each vertex $v \in V(H)$ is at most $3/2$.\label{claim:chargeOfVertex}\end{claim}

\begin{proof} Let $v \in V(H)$. The charge of $v$ can increase only in times when some simple cycle or clique that contains $v$ was chosen by the procedure. Let $r \ge 0$ be the number of such cycles/cliques, and let $(S_1, ..., S_r)$ be a (possibly empty) list containing all these cycles/cliques ordered by the time they were chosen. Each time a simple cycle or a clique is chosen, at least one edge is removed from it. Therefore, all the simple cycles and cliques that where chosen by the procedure must be different from each other, and in particular the cycles/cliques in $(S_1, ..., S_r)$. We claim that all these cycles/cliques are edge-disjoint. Let $1 \le i < j \le r$. We distinguish between the following cases: \smallskip

Case 1: $S_i$ is a simple cycle of length $5$. If $S_j$ is a $4$-clique then by Lemma \ref{lem:cliqueIsDisjoint} (\ref{lem:case1}) it cannot share an edge with $S_i$. Otherwise, $S_j$ must be a simple cycle of length at most $5$ that is different from $S_i$ and so by Lemma \ref{lem:cycleIsDisjoint} it cannot share an edge with $S_i$.\smallskip

Case 2: Both $S_i$ and $S_j$ are $4$-cliques. We must have $S_i \ne S_j$ and by Lemma \ref{lem:cliqueIsDisjoint} (\ref{lem:case2}) we get that $S_i$ and $S_j$ are edge-disjoint.\smallskip

Case 3: $S_i$ is a $4$-clique and $S_j$ is a simple cycle of length $3$. When $S_i$ was chosen, enough edges were removed from $H$ so that $H[V(S_i)]$ would not contain any simple cycle. Therefore, we cannot have $V(S_j) \subseteq V(S_i)$ and so by Lemma \ref{lem:cliqueIsDisjoint} (\ref{lem:case3}) we get that $S_i$ and $S_j$ are edge-disjoint.\smallskip

Case 4: Both $S_i$ and $S_j$ are simple cycles of length $3$. Assume towards a contradiction that they are not edge-disjoint. Let $H_3$ be the graph resulting from $H$ at the end of step \ref{proc:step3} and before the start of step \ref{proc:step4}. $H_3$ must contain both $S_i$ and $S_j$, and since $H_3$ does not contain any $4$-clique (as in step \ref{proc:step3}, we kept removing edges until no $4$-clique remains), the graph $H_3[V(S_i) \cup V(S_j)]$ cannot be isomorphic to a complete graph. But, $S_i$ and $S_j$ also exist is $H$ (as $H_3$ is a subgraph of $H$), and so by Lemma \ref{lem:edgeDisCycles} we must have that $H[V(S_i) \cup V(S_j)]$ is isomorphic to a complete graph (as $S_i$ and $S_j$ are two different, not edge-disjoint simple cycles of length $\le 10$). It follows that exactly one edge must have been removed from $H[V(S_i) \cup V(S_j)]$ at some point during steps \ref{proc:step2} and \ref{proc:step3}. This is possible only if some simple cycle of length $5$ or $4$-clique different from $H[V(S_i) \cup V(S_j)]$ shares an edge with $H[V(S_i) \cup V(S_j)]$, a contradiction to Lemma \ref{lem:cliqueIsDisjoint} (\ref{lem:case1}, \ref{lem:case2}).\smallskip

Now, we show that the charge of $v$ can be at most $3/2$. We consider separately the case in which one of $\{S_1, ..., S_r\}$ is a $4$-clique, and the case in which none of them is a $4$-clique. In the first case, $v$ belongs to some $4$-clique $S \in \{S_1, ..., S_r\}$. Since the degree of $v$ is at most $6$ and all the cycles/cliques in $\{S_1, ..., S_r\}$ are edge disjoint, we get that $|\{S_1, ..., S_r\}| \le 2$ (since if $|\{S_1, ..., S_r\}| \ge 3$ then $v$ should have $3$ distinct neighbors for $S$ and at least $2$ distinct neighbors for each of the remaining cycles/cliques, that is, at least $4$ additional neighbors which is impossible as $d(v,H) \le 6$). It follows that $v$ gets a charge of $3/4$ for $S$ and possibly another charge of at most $3/4$, thus its charge at the end can be at most $3/4 + 3/4 = 3/2$. In the second case, there is no $4$-clique in $\{S_1, ..., S_r\}$. This means that every cycle/clique in $\{S_1, ..., S_r\}$ is a simple cycle of length $3$ or $5$. Since these cycles are edge-disjoint and the degree of $v$ is at most $6$, we get that $|\{S_1, ..., S_r\}| \le 3$ (as $v$ should have $2$ distinct neighbors for each such cycle). It follows that the charge of $v$ can be at most $1$ as each such cycle can contribute at most $1/3$ to the charge of $v$.\end{proof}

It is easy to see that each time a set of edges $E'$ is removed from the graph, a corresponding charge of $|E'|$ is added to the total charge of the vertices. Therefore, by the end of the procedure, the total number of removed edges is indeed equal to the total charge of the vertices which by claim \ref{claim:chargeOfVertex} can be at most $1.5|V(H)|$. We prove the following claim:

\begin{claim} The resulting graph $H'$ does not contain any simple cycle of length at most $5$.\label{claim:noShortCycles}\end{claim}

\begin{proof} Clearly, $H'$ cannot contain any simple cycle of length $3$ as in step \ref{proc:step4} we kept removing edges until no such cycle remained. Similarly, $H'$ cannot contain any simple cycle of length $5$ as in step \ref{proc:step2} we kept removing edges until no such cycle remained. It left to show that $H'$ cannot contain any simple cycle of length $4$. Assume towards a contradiction that $H'$ contains such a cycle $C$. We cannot have that $H'[V(C)]$ is isomorphic to a complete graph as in step \ref{proc:step3} we kept removing edges until no $4$-clique remained. Since $C$ exists in $H$, we must have that $H[V(C)]$ is isomorphic to a complete graph as $H$ does not contain any removable cycle of length $4$. This is possible only if some edges were removed from $H[V(C)]$ during steps \ref{proc:step2} and \ref{proc:step3}. The procedure could not choose $C$ at step \ref{proc:step3} as enough edges would have been removed from the graph so that $H'[V(C)]$ could not contain any simple cycle. It follows that there must be some other cycle/clique different from $C$ that share an edge with $C$ and was removed during step \ref{proc:step2} or \ref{proc:step3}, in contradiction to Lemma \ref{lem:cliqueIsDisjoint} (\ref{lem:case1}, \ref{lem:case2}).\end{proof}

Now, the resulting graph $H'$ cannot contain a simple cycle of length $6$ as each such cycle $C$ must also be in $H$, and since $H$ does not contain any removable cycle of length $6$, we get that $H[V(C)]$ must be isomorphic to a complete graph which is impossible as $H$ is planar. It follows that $H'$ is a planar graph with no simple cycle of length at most $6$, and so the number of edges in $H'$ is at most $\frac{7}{7-2}|V(H')| = 1.4|V(H)|$ (by claim \ref{claim:numberOfEdges}). Since $H'$ was obtained from $H$ by removing at most $1.5|V(H)|$ edges, we get that $H$ contains at most $1.4|V(H)| + 1.5|V(H)| = 2.9|V(H)|$ edges.\end{proof}

\subsection{The algorithm}

We assume that communication network $G = (V,E)$ is a planar graph on $n$ vertices where each vertex $v \in V$ has a unique identifier $\text{ID}(v)$ taken from $\{1,...,n^{c'}\}$ (for some constant integer $c' \ge 1$). At the beginning each vertex $v \in V$ knows only the value of $n$ and its own identifier. In the first part, each vertex $v \in V$ executes in parallel the following algorithm:\medskip \smallskip

\begin{algorithm}[H]\label{Alg6ColorPartition}\small
	\caption{partitions the vertices of the graph}
\SetInd{1em}{0em}
$\text{L}[v] \gets (\perp, \perp)$\\

	\For {\upshape $i \gets 1$ \textbf{to} $1 + 700\lceil \log_{2}n \rceil$} { \label{alg:outer6}
 		\tcp{Step 1: removing all removable cycles of length $\le 10$ and degree $\le 6$}
	            Collect the labeled neighborhood of $v$ up to distance $6$, and let $N(v)$ be the subgraph induced by all the vertices $u$ with $\text{L}[u] = (\perp, \perp)$ whose distance from $v$ is at most $6$. \\
             \If{\upshape $v$ belongs to some removable cycle $C$ in $N(v)$ of length $\le 10$ and degree $\le 6$}  {$\text{L}[v] \gets (i,1)$ \\
                    $\text{key}[v] \gets \{\text{ID}(u) \mid u \in V(C)\}$} \smallskip

            \tcp{Step 2: removing all vertices of degree less than $6$}
            Collect the labeled neighborhood of $v$ up to distance $1$, and let $N(v)$ be the subgraph induced by all the vertices $u$ with $\text{L}[u] = (\perp, \perp)$ whose distance from $v$ is at most $1$. \\
            \If{\upshape $\text{L}[v] = (\perp, \perp)$ \textbf{and} $v$ is of degree less than 6 in $N(v)$}
                {\label{removeLowDeg6} $\text{L}[v] \gets (i,2)$ \\$\text{key}[v] \gets \{\text{ID}(v)\}$}
    }

\end{algorithm}\medskip \smallskip

The next lemma is the analogue of Lemma \ref{lem:numberOfVerticesDeactivates} from the previous section.

\begin{lemma} For each $i \in \{1,...,1+700\lceil\log_{2}n\rceil\}$, if $A$ and $B$ respectively denote the sets of all the vertices that were active at the start of the outer loop's $i$-th iteration and at its end, then $|B| \le 0.999|A|$.\label{lem:numberOfVerticesDeactivates6}\end{lemma}

\begin{proof}Let $L_A \subseteq A$ and $H_A \subseteq A$ be the sets of all vertices of degree less than $6$ and greater than $6$ in $G[A]$, respectively. Since each vertex who becomes inactive at some point of the algorithm's execution remains inactive until the end, it is enough to show that at least $0.001|A|$ of the vertices in $A$ become inactive by the end of the $i$-th iteration. This is clearly the case when $|L_A| \ge 0.001|A|$ as each vertex in $L_A$ who remained active after step 1, must become inactive by the end of step 2. So it remains to show that this is also the case when $|L_A| < 0.001|A|$.

Let $U$ be the set of all vertices that remained active at the end of the first step of the $i$-th iteration. We divide the edges of the graph $G[U]$ into two disjoint subsets $E_{1}$ and $E_{2}$ where $E_{1}$ contains all the edges in $G[U]$ that are incident to some vertex in $H_A$, and $E_{2}$ contains all the other edges. Since $G[A]$ is a planar graph, we have 
$|E(G[A])| \le 3|A|$ (by setting $g = 3$ in Claim \ref{claim:numberOfEdges}). It follows that $\Sigma_{v \in A}d(v, G[A]) = 2|E(G[A])| \le 6|A|$ and so $6|A| \ge \Sigma_{v \in A}d(v, G[A]) \ge \Sigma_{v \in A \setminus L_A} d(v, G[A]) = \Sigma_{v \in A \setminus L_A} (d(v, G[A]) - 6) + 6|A \setminus L_A|$. This implies that $6|L_A| = 6|A| - 6|A \setminus L_A| \ge \Sigma_{v \in A \setminus L_A} (d(v, G[A]) - 6) \ge |H_A|$, where the last inequality follows from the fact that each vertex in $H_A$ contributes at least $1$ to the sum $\Sigma_{v \in A \setminus L_A} (d(v, G[A]) - 6)$. Clearly, $\Sigma_{v \in H_A}(d(v, G[A]) - 6) = \Sigma_{v \in A \setminus L_A}(d(v, G[A]) - 6)$ and so $\Sigma_{v \in H_A}d(v, G[A]) = \Sigma_{v \in H_A}(d(v, G[A]) - 6) + 6|H_A| = \Sigma_{v \in  A \setminus L_A}(d(v, G[A]) - 6) + 6|H_A| \le 6|L_A| + 36|L_A| = 42|L_A| < 0.042|A| < 0.05|A|$. We conclude that $|E_{1}| \le \Sigma_{v \in H_A}d(v, G[A]) < 0.05|A|$. Now, let $G'$ be the subgraph induced by the edges in $E_{2}$, that is the graph whose set of vertices $v$ is all vertices in $V$ such that $v$ is incident to at least one edge in $E_2$, and whose set of edges is $E_2$. Since each vertex in $G'$ has degree at most $6$ in $G[A]$, and $G'$ is a subgraph of $G[A]$, we get that the maximum degree of $G'$ is at most $6$. Moreover, $G'$ does not contain any removable cycle of length at most $10$ as any such cycle would have been removed at step $1$. It follows that $G'$ is a planar graph of degree at most $6$ that does not contain any removable cycle of length at most $10$, and so from Lemma \ref{lem:numberOfEdgesWithoutRemovableCycles} we get that $G'$ contains at most $2.9|V(G')|$ edges. We conclude that $|E_2| = |E(G')| \le 2.9|V(G')| \le 2.9|A|$, and so $G[U]$ contains at most $|E_{1}| + |E_{2}| < 2.95|A|$ edges. Now, let $L_U$ be the set of all vertices of degree less than $6$ in $G[U]$. We have $5.9|A| \ge 2(|E_{1}| + |E_{2}|) = \Sigma_{v \in U}d(v, G[U]) \ge \Sigma_{v \in U \setminus L_U}d(v, G[U]) \ge 6(|U| - |L_U|) = 6(|A| - |A \setminus U| - |L_U|)$ and so $6(|A \setminus U| + |L_U|) \ge 0.1|A|$. It follows that $|A \setminus U| + |L_U| \ge 0.01|A|$ and so $|B| \le 0.99|A|$. 
\end{proof}

Let $\beta = 1 + 700\lceil \log_{2}n \rceil$ and let $S = \{1,...,\beta\} \times \{1,2\}$. The above algorithm partitions the vertices of the graph into $|S|$ disjoint subsets $\{H_{i,j}\}_{(i,j) \in S}$ where each $H_{i,j}$ contains all the vertices $u$ with $\text{L}[u] = (i,j)$. Given a set of vertices $U$, we define $N(U)$ to be the set containing all the vertices $u$ whose distance in $G$ from some vertex with identifier in $U$ is at most $1$. For each $i \in \{1, ..., \beta\}$, we denote by ${\cal G}_{i} = ({\cal V}_{i},{\cal E}_{i})$ the graph whose nodes and edges are $\{\text{key}[u] \mid u \in H_{i,1}\}$ and $\{\{\text{key}[u], \text{key}[v]\}\mid \{u,v\} \subseteq H_{i,1} \text{\; and \;} \text{key}[u] \ne \text{key}[v] \text{\; and \;} N(\text{key}[u]) \cap \text{key}[v] \ne \emptyset  \}$, respectively.

\begin{lemma}  For each $i \in \{1, ..., \beta\}$, the maximum degree of the graph ${\cal G}_{i}$ is less than $6^{11}$. \label{lem:degreeSuperGraph6}\end{lemma}
\begin{proof} Let $i \in \{1, ..., \beta\}$ and $v \in H_{i,1}$. The degree of $\text{key}[v]$ in ${\cal G}_{i}$ is equal to the size of the set $\{\text{key}[u] \mid u \in H_{i,1} \text{\; and \;} N(\text{key}[u]) \cap \text{key}[v] \ne \emptyset  \text{\; and \;} \text{key}[u] \ne \text{key}[v]\}$.
In other words, the neighbors of $\text{key}[v]$ are nodes that represent cycles $C'$ in ${\cal G}_{i}$ that either have a common vertex with the cycle $C$ that $\text{key}[v]$ represents or there is an edge such that one endpoint of this edge is in $C$ and the other is in $C'$.
It is not hard to verify that the size of this set is bounded by the size of $A' = \{u \in H_{i,1} \mid N(\text{key}[u]) \cap \text{key}[v] \ne \emptyset\}$.

Let $u \in A'$, let $C_u$ be the cycle that $\text{key}[u]$ represents and $C_v$ be the cycle that $\text{key}[v]$ represents.
By definition, we have $V(C_u) \cap V(C_v) \neq \emptyset$ or there is an edge $\{x,y\} \in E$ such that $x \in V(C_u)$ and $y \in V(C_v)$.
Since both cycles  $C_u$ and $C_v$ are of length $\le 10$ then it is not hard to verify that the distance between $u$ and $v$ in the induced graph of
$V(C_u) \cup V(C_v)$ is at most $11$.
Moreover, all vertices of $V(C_u) \cup V(C_v)$ belong to $H_{i,1}$.
We get that the distance between $u$ and $v$ in the induced graph $G[H_{i,1}]$ is at most $11$.
Note that the maximum degree in $G[H_{i,1}]$ is at most $6$.
Straight forward calculations show that for a vertex $v'$ in a graph of degree at most $6$ there could at most $6^{11}$ vertices at distance at most $11$ from it (we didn't try to optimize constants). Hence there could be at most $6^{11}$ such vertices $u$ in the set $A'$.\end{proof}

For each $i \in \{1, ..., \beta\}$, we want to compute a proper $7$-coloring $\varphi_{i}$ of the graph $G[H_{i,2}]$ so that each vertex $v$ with $\text{L}[v] = (i,2)$ knows the value of $\varphi[v] = \varphi_{i}(v)$. This can be done in $O(\log^{*}n)$ by executing in parallel a $\Delta+1$-coloring algorithm (e.g. \cite{Goldberg87}) on the graph $G[H_{i,2}]$. Next, we want to compute a proper coloring $\varphi_{i}:{\cal V}_{i} \to \{1,2,...,6^{11}\}$ of the graph ${\cal G}_{i}$ so that each vertex $v$ with $\text{L}[v] = (i,1)$ knows the value of $\varphi[v] = \varphi_{i}(\text{key}[v])$. This can be done in $O(\log^{*}n)$ rounds by simulating in parallel a $6^{11}$-coloring algorithm for each ${\cal G}_{i}$ on the graph $G$ (in a similar manner to the $4$-coloring algorithm in the previous section). \medskip \smallskip

\begin{algorithm}[H]\label{Alg6ColorSuperGraph}\small
	\caption{Give a color $\varphi[u]$ for every vertex $u$. The colors $\varphi[u]$ are supposed to synchronize between the different vertices in each $H_{i,j}$ for $i \in \{1, ..., \beta\}$ and $1 \leq j \leq 2$}
\SetInd{1em}{0em}
Let $i,j$ be the the indices such that $L(v) = (i,j)$.\\
If $j=2$ then execute the $\Delta+1$-coloring algorithm of \cite{Goldberg87} in the graph $G[H_{i,j}]$ and set $\varphi[v]$ to be the color assigned to $v$.\\
If $j=1$  then execute the $\Delta+1$-coloring algorithm of \cite{Goldberg87} in the super graph ${\cal G}_i$ (by simulating the graph ${\cal G}_i$ such that every cycle $\text{key}[u]$ is simulated by the vertex with highest ID in $\text{key}[u]$) and set $\varphi[v]$ to be the color assigned to $\text{key}[v]$.\label{alg2:col6}

\end{algorithm}

\medskip \smallskip

Now, in order to find a $6$-coloring of the whole graph each vertex $v \in V$ executes in parallel the following algorithm:
\medskip \smallskip

\begin{algorithm}[H]\label{Alg6ColorGraph}\small
	\caption{color the vertices of the graph}
\SetInd{1em}{0em}
$\text{C}[v] \gets \;\perp$\\
Collect the value of $(\text{key}[u], \varphi[u])$ from each vertex $u$ with $\text{L}[v] = \text{L}[u]$ whose distance from $v$ is at most $5$, and let $\text{P}[v]$ be the set containing all the pairs (among the collected ones) whose key contains the ID of $v$.\\ \label{line:collectCycle6}
Choose $(\text{k},\text{color}) \in \text{P}[v]$ with maximum color, and set $\text{key}_{\text{new}}[v] \gets \text{k}$ and $\varphi_{\text{new}}[v] \gets \text{color}$.\label{line:chooseCycle6}

	\For {\upshape $i \gets 1 + 700\lceil \log_{2}n \rceil$ \textbf{downto} $1$} { \label{line:outer6}

\tcp{Step 1: coloring all vertices $v \in V$ with $\text{L}[v] = (i,2)$}
		    \For {\upshape $k \gets 1$ \textbf{to} $7$} { \label{line:innerStep1_6}
		Collect the value of $\text{C}[u]$ from each neighbor $u$ of $v$, and let $\text{Colors}$ be the set containing all these values.

             \If{\upshape $\text{L}[v] = (i,2)$ \textbf{and} $\varphi_{\text{new}}[v] = k$}  { \label{line:innerCond1_6}
		Choose a color $x \in \{1,2,3,4,5,6\} \setminus \text{Colors}$, and set $\text{C}[v] \gets x$.\label{line:ColorVertex6}
      }}

                   \smallskip\tcp{Step 2: coloring all vertices $v \in V$ with $\text{L}[v] = (i,1)$}
		    \For {\upshape $k \gets 1$ \textbf{to} $6^{11}$} { \label{alg2:inner6}
		Collect the labeled neighborhood of $v$ up to distance $11$ (along with the value of $C$ of each vertex in that neighborhood), and let $N_{i,k}(v)$ be the subgraph induced by all the vertices whose ID belongs to $N(\text{key}_{\text{new}}[v])$.

             \If{\upshape $\text{L}[v] = (i,1)$ \textbf{and} $\varphi_{\text{new}}[v] = k$}  {\label{line:innerCond2_6}
Find a proper $6$-coloring of $N_{i,k}(v)$ which is consistent with the vertices that are already colored (by using a deterministic algorithm), and set $\text{C}[v]$ accordingly.  \label{line:ColorRemovable_6}

}}}

\end{algorithm}

\medskip \smallskip

For each $i \in \{0, ..., \beta\}$, we let $A_i$ be the set of all vertices $u \in V$ with $\text{L}[u] \in \{j \mid i < j \le \beta\} \times \{1,2\}$. Our aim is to show that $A_{0} = V$  is properly $6$-colored by the end of the algorithm. We do this by showing that for all $i \in \{1,...,\beta\}$ the following holds: if $A_{i}$ is properly colored at the start of the $i$-th iteration of the loop at line \ref{line:outer6}, then $A_{i-1}$ is properly colored by the end of this iteration. This claim easily implies what we want as $A_{\beta} = \emptyset$ and so we can apply this claim iteratively to get that $A_{0}$ is properly $6$-colored by the end. The proof of this claim is similar to the proof in the $4$-coloring algorithm for triangle-free planar graphs, so we only state the required lemmas and omit the proofs.

\begin{lemma}  For all $i \in \{1, ..., \beta\}$, if at the start of the $i$-th iteration of the loop at line \ref{line:outer6}, we have that $A_i$ is properly colored and $V \setminus A_i$ is not colored, then by the end of step $1$ of the $i$-th iteration of this same loop, we get that $A_i \cup H_{i,2}$ is properly colored and $V \setminus (A_i \cup H_{i,2})$ is not colored.\end{lemma}

\begin{lemma} For any $(i,k) \in \{1, ..., \beta\} \times \{1, ..., 6^{11}\}$ and $v \in V$, if $\text{L}[v] = (i,1)$ and $\varphi_{\text{new}}[v] = k$ then at the start of the $k$-th iteration of the loop in line \ref{alg2:inner6} during the $i$-th iteration of the loop in line \ref{line:outer6} we must have $\text{C}[u] = \; \perp$ for every vertex whose ID belongs to $\text{key}_{\text{new}}[v]$.\end{lemma}

\begin{lemma}  For all $i \in \{1, ..., \beta\}$, if at the start of the second step of the $i$-th iteration of the loop at line \ref{line:outer6}, we have that $A_i \cup H_{i,2}$ is properly colored and $V \setminus (A_i \cup H_{i,2})$ is not colored, then by the end of step $2$ of the $i$-th iteration of this same loop, we get that $A_i \cup H_{i,2} \cup H_{i,1}$ is properly colored and $V \setminus (A_i \cup H_{i,2} \cup H_{i,1})$ is not colored.\end{lemma}

\begin{corollary}  For all $i \in \{1, ..., \beta\}$, if $A_{i}$ is properly colored at the start of the $i$-th iteration of the loop at line \ref{line:outer6}, then $A_{i-1}$ is properly colored by the end of the $i$-th iteration of this loop.\end{corollary}

\begin{lemma}The total number of communication rounds is $O(\log{n})$, and it is possible to implement the algorithm in such a way that the size of any message sent by any vertex during the algorithm's execution is $O(\log{n})$. \end{lemma}

\section{Lower Bounds} \label{sec:lowerBounds}
In this section, we prove that any distributed algorithm (possibly randomized) that computes a proper $4$-coloring (resp. $3$-coloring) for every planar graph (resp. outerplanar graph) on $n$ vertices requires $\Omega(n)$ communication rounds.

\subsection{Lower bound for 4-coloring planar graphs}

	We define a sequence of planar graphs $\{G_k\}_{k \in \Z_{\ge}}$ inductively as follows: the base graph $G_0$ contains two different vertices and a single edge connecting them, that is, $V(G_0) = \{v_{(0,1)},v_{(0,2)}\}$ and $E(G_0) = \{\{v_{(0,1)},v_{(0,2)}\}\}$. For each $i \in \Z_{>}$, the graph $G_i$ is defined by using the graph $G_{i-1}$ as follows: $V(G_i) = V(G_{i-1}) \cupdot A_i$ and $E(G_{i}) = E(G_{i-1}) \cupdot B_i$ where $A_{i} = \bigcup_{j \in \{1,2\}}\left\{a_{(i,j)},b_{(i,j)},c_{(i,j)},v_{(i,j)}\right\}$ and  $B_{i} = \bigcup_{j \in \{1,2\}}\{\{a_{(i,j)},b_{(i,j)}\},$\\$\{b_{(i,j)},c_{(i,j)}\},\{c_{(i,j)},a_{(i,j)}\},\{a_{(i,j)},v_{(i,j)}\},\{b_{(i,j)},v_{(i,j)}\},\{c_{(i,j)},v_{(i,j)}\},\{a_{(i,j)},v_{(i-1,j)}\},\{b_{(i,j)},v_{(i-1,j)}\},$\\$\{c_{(i,j)},v_{(i-1,j)}\}\}$ (see Figure \ref{planarLowerBound} for illustration).
	
\begin{figure}[ht!]
	\centering
	\includegraphics{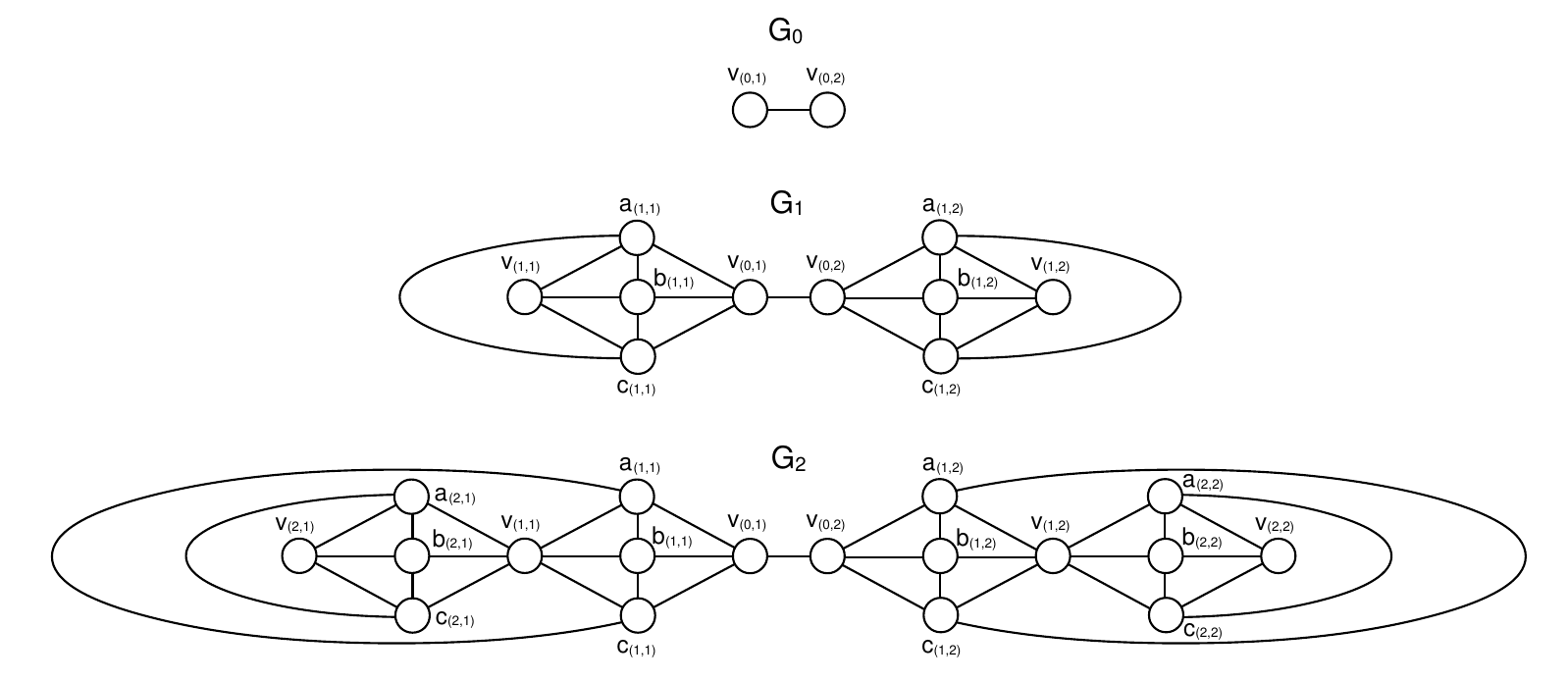}
	\caption{An illustration of the graphs $G_{0}$, $G_{1}$ and $G_{2}$. \label{planarLowerBound}}
\end{figure}

\medskip

The following key lemma shows that in any proper $4$-coloring of $G_k$, there are vertices at distance $\Omega(|V(G_k)|)$ that must have the same color.

\begin{lemma} For any non-negative integer $k$, if $\varphi$ is a proper $4$-coloring of $G_k$ then $\varphi(v_{(k,1)}) = \varphi(v_{(0,1)})$ and $\varphi(v_{(k,2)}) = \varphi(v_{(0,2)})$.\label{lem:distanceColor}\end{lemma}

\begin{proof} By induction on $k$. The base case ($k = 0$) is trivial. Assume that the claim holds for some $k \in \Z_{\ge}$ and prove it for $k+1$. Let $\varphi$ be some proper $4$-coloring of $G_{k+1}$. We claim that $\varphi(v_{(k+1,1)}) = \varphi(v_{(k,1)})$. Indeed, consider the set $S = \{v_{(k,1)},a_{(k+1,1)},b_{(k+1,1)},c_{(k+1,1)},v_{(k+1,1)}\}$. 
Note that, by construction, the graph $G_{k+1}[S\setminus\{v_{(k+1,1)}\}]$ is isomorphic to a complete graph, and so the vertices
$v_{(k,1)},a_{(k+1,1)},b_{(k+1,1)}$ and $c_{(k+1,1)}$ must get a different color each.
In addition, the graph $G_{k+1}[S\setminus\{v_{(k,1)}\}]$  is isomorphic to a complete graph, and therefore
the vertices
$v_{(k+1,1)},a_{(k+1,1)},b_{(k+1,1)}$ and $c_{(k+1,1)}$ must also get a different color each. Since $\varphi$ is a $4$-coloring of $G_{k+1}$, we must have $\varphi(v_{(k,1)}) = \varphi(v_{(k+1,1)})$.
A similar argument shows that $\varphi(v_{(k+1,2)}) = \varphi(v_{(k,2)})$. Now, $G_k$ is a subgraph of $G_{k+1}$ and so the restriction of $\varphi$ to the set $V(G_k)$ is a proper $4$-coloring of $G_k$. Thus, by the induction hypothesis, we get that $\varphi(v_{(k,1)}) = \varphi(v_{(0,1)})$ and $\varphi(v_{(k,2)}) = \varphi(v_{(0,2)})$. It follows that $\varphi(v_{(k+1,1)}) = \varphi(v_{(0,1)})$ and $\varphi(v_{(k+1,2)}) = \varphi(v_{(0,2)})$.\end{proof}

The following trivial observation claims that the shortest distance from $v_{(0,j)}$ to $v_{(k,j)}$ for every $j \in \{1,2\}$ is large and can easily be verified by construction.

\begin{observation} For any non-negative integer $k$ and $j \in \{1,2\}$, the shortest distance between $v_{(0,j)}$ and $v_{(k,j)}$ in the graph $G_{k}$ is equal to $2k$. \label{lem:longDis}\end{observation}

%
\begin{definition}
	Given a graph $G$, a non-negative integer $t$, and a vertex $v$ in $G$, we define $B_{t}(G,v)$ to be the set containing all the vertices in $G$ whose shortest distance from $v$ is at most $t$.
\end{definition}

\begin{lemma} For any $k,t \in \Z_{\ge}$, if $t < k$ then $B_{t}(G_{k},v_{(k,j_{1})}) \cap B_{t}(G_{k},v_{(0,j_{2})}) = \emptyset$ for every $j_{1},j_{2} \in \{1,2\}$.\label{lem:disBalls}\end{lemma}

\begin{proof} We prove the case in which $j_{1} = 1$ (the proof of the other case is symmetric). Assume towards a contradiction that $B_{t}(G_{k},v_{(k,1)}) \cap B_{t}(G_{k},v_{(0,1)}) \ne \emptyset$. This implies that there exists a path in $G_{k}$ from $v_{(0,1)}$ to $v_{(k,1)}$ of length $\le 2t < 2k$ in contradiction to Observation \ref{lem:longDis}. Similarly, if we assume that $B_{t}(G_{k},v_{(k,1)}) \cap B_{t}(G_{k},v_{(0,2)}) \ne \emptyset$, we would get that there exists in $G_{k}$ a path from $v_{(0,2)}$ to $v_{(k,1)}$ of length $< 2k$ which would imply that there exists a path from $v_{(0,1)}$ to $v_{(k,1)}$ of length $< 2k$ (because from the construction of $G_{k}$ it follows that every path from $v_{(0,2)}$ to $v_{(k,1)}$ must pass through $v_{(0,1)}$) which is impossible.\end{proof}

\begin{lemma} For every positive integer $k$, any deterministic distributed algorithm that finds a proper $4$-coloring of $G_k$ requires at least $k = \frac{1}{8}(|V(G_k)|-2)$ rounds.\label{lem:lowerBoundDet}\end{lemma}

\begin{proof} Let $k \in \Z_{>}$ and assume towards a contradiction that there exists a deterministic distributed algorithm $A$ that finds a proper $4$-coloring of $G_k$ using at most $t < k$ rounds. Let $\psi$ be some legal labeling of $V(G_k)$, and denote by $\varphi:V(G_k) \to \{1,2,3,4\}$ the output of $A$ when it is invoked on the input $(G_k,\psi)$. From the assumption that $A$ is correct, it follows that $\varphi$ is a proper $4$-coloring of $G_k$, and thus from Lemma \ref{lem:distanceColor} we get that $\varphi(v_{(k,1)}) = \varphi(v_{(0,1)}) = x$ and $\varphi(v_{(k,2)}) = \varphi(v_{(0,2)}) = y$. Moreover, since $v_{(0,1)}$ and $v_{(0,2)}$ are connected by an edge in $G_{k}$, we must have $x \ne y$. Now, we define a new labeling $\psi'$ of $V(G_k)$ by $$\psi'(z_{(i,j)}) = \begin{cases}
		\psi(z_{(i,3-j)}) & \text{if } z_{(i,j)} \in B_{t}(G_k, v_{(k,1)}) \cup B_{t}(G_k, v_{(k,2)}) \\
		\psi(z_{(i,j)}) & \text{otherwise}
	\end{cases}  $$
	for every $z \in \{a,b,c,v\}$ and denote by $\varphi'$ the output of $A$ when invoked on the input $(G_{k},\psi')$.
By Lemma \ref{lem:disBalls}, $(B_{t}(G_{k},v_{(k,1)}) \cup B_{t}(G_{k},v_{(k,2)})) \cap B_{t}(G_{k},v_{(0,1)}) = \emptyset$ and so $\psi'(w) = \psi(w)$ for every $w \in B_{t}(G_{k},v_{(0,1)})$ which implies that $\varphi'(v_{(0,1)}) = \varphi(v_{(0,1)}) = x$. But, $\varphi'(v_{(k,1)}) = \varphi(v_{(k,2)}) = y \ne x = \varphi'(v_{(0,1)})$ and so, by Lemma \ref{lem:distanceColor}, $\varphi'$ is not a proper $4$-coloring of $G_k$. Contradiction to the assumption that $A$ is correct.\end{proof}

The next lemma shows that the same lower bound holds for randomized distributed algorithms as well.

\begin{lemma} For any positive integer $k$, any randomized distributed algorithm that finds a proper $4$-coloring of $G_k$ with success probability $>1/2$ requires at least $k$ rounds.\label{lem:lowerBoundRand}\end{lemma}

\begin{proof}Any randomized algorithm that runs in the current model, can be simulated by an algorithm that runs in the following model: A random string $\sigma$ is first generated by some external random source and announced to all vertices in the graph. From that point on, the algorithm proceeds in the usual way. That is, the color of each vertex is just a function of $\sigma$ and its labeled $t$-neighborhood (where $t$ is the number of rounds). Now, it is easy to see that any randomized algorithm that runs in the latter model and finds a proper $4$-coloring of $G_k$ using at most $t < k$ rounds, must fail on at least one of the labeling $\psi$ or $\psi'$ (as were defined in Lemma \ref{lem:lowerBoundDet}) for any given $\sigma$. This implies that there exists a labeling for which the algorithm fails with probability at least $1/2$. 
\end{proof}

We conclude the following:

\begin{corollary} Any distributed algorithm that computes (possibly with high probability) a proper $4$-coloring for every planar graph on $n$ vertices requires $\Omega(n)$ rounds. \label{col:lowerBound4Col}\end{corollary}

\subsection{Lower bound for 3-coloring outerplanar graphs}
        We define a sequence of outerplanar graphs $\{G_k\}_{k \, \in\, \Z_{\ge}}$ inductively as follows: the base graph $G_0$ contains two different vertices and a single edge connecting them, that is, $V(G_0) = \{v_0,u_0\}$ and $E(G_0) = \{\{v_0,u_0\}\}$. For each $i \in \Z_{>}$, the graph $G_i$ is defined to be the graph obtained from $G_{i-1}$ by connecting a new and different triangle to each endpoint vertex of $G_{i-1}$ (that is, to each vertex of degree at most $2$ in $G_{i-1}$), where each such triangle is connected to its matching endpoint $v$ by connecting $v$ to two different vertices in that triangle, that is, $V(G_i) = V(G_{i-1}) \cupdot A_i$ and $E(G_i) = E(G_{i-1}) \cupdot B_i$ where $A_i = \{a_i,b_i,c_i,d_i,v_i,u_i\}$ and $B_i = \{\{a_i,b_i\},\{c_i,d_i\},\{a_i,v_i\},\{b_i,v_i\},\{c_i,u_i\},\{d_i,u_i\},\{a_i,v_{i-1}\},\{b_i,v_{i-1}\},\{c_i,u_{i-1}\},\{d_i,u_{i-1}\}\}$ (see Figure \ref{fig:outerplanar} for illustration).
\begin{figure}[ht!]
	\centering
	\includegraphics[trim={0.08cm 0.2cm 0.1cm 0.25cm},clip]{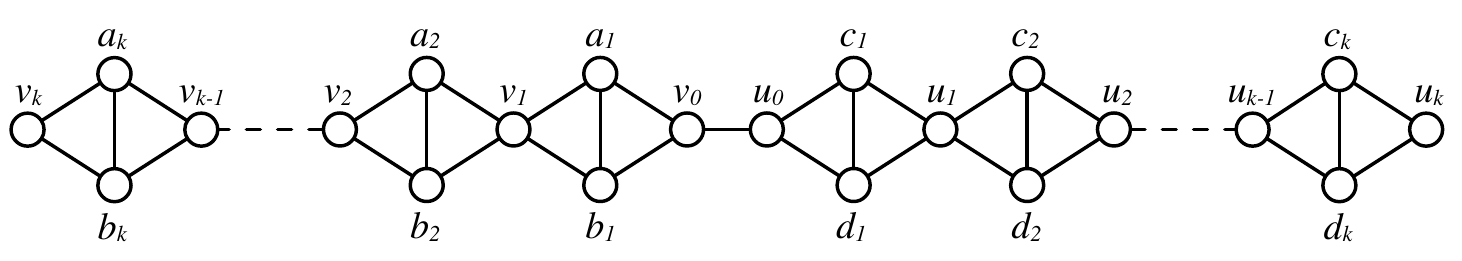}
	\caption{An illustration of the graph $G_k$ for some positive integer $k$.} \label{fig:outerplanar}
\end{figure}


The proof proceeds in the same way as before, and so we only state the important lemmas:

\begin{lemma} For any $k \in \Z_{\ge}$ and any proper $3$-coloring $\varphi$ of $G_k$, we have $\varphi(v_k) = \varphi(v_0)$ and $\varphi(u_k) = \varphi(u_0)$.\label{lem:distanceColorOuterplanar}\end{lemma}


\begin{lemma} For every positive integer $k$, any deterministic distributed algorithm that finds a proper $3$-coloring of $G_k$ requires at least $k = \frac{1}{6}(|V(G_k)|-2)$ rounds.\label{lem:lowerBoundDet3}\end{lemma}

\begin{lemma} For every positive integer $k$, any randomized distributed algorithm that finds a proper $3$-coloring of $G_k$ with success probability $>1/2$ requires at least $k$ rounds.\label{lem:lowerBoundRand3}\end{lemma}

We conclude the following:

\begin{corollary} Any distributed algorithm that computes (possibly with high probability) a proper $3$-coloring for every outerplanar graph on $n$ vertices requires $\Omega(n)$ rounds. \label{col:lowerBound3Col}\end{corollary}

By using the same techniques (that we used in the previous coloring algorithms) and a similar analysis, it is possible to get a deterministic algorithm that finds a proper $4$-coloring for every outerplanar graph on $n$ vertices in time $O(\log{n})$.

\bibliographystyle{plain}
\bibliography{ref}

\end{document}